\documentclass[sigconf, nonacm]{acmart}
\AtBeginDocument{%
  }

\setcopyright{none}
\copyrightyear{2018}
\acmYear{2018}
\acmDOI{XXXXXXX.XXXXXXX}
\acmConference[Conference acronym 'XX]{Make sure to enter the correct
 conference title from your rights confirmation email}{June 03--05,
 2018}{Woodstock, NY}

\acmBooktitle{Woodstock '18: ACM Symposium on Neural Gaze Detection, June 03--05, 2018, Woodstock, NY}
\acmISBN{978-1-4503-XXXX-X/2018/06}

\usepackage{xspace}
\usepackage{enumitem}
\usepackage{multirow}
\usepackage{makecell}
\usepackage{subcaption}
\usepackage{graphicx}
\usepackage[ruled,linesnumbered,noend]{algorithm2e}
\usepackage{pifont}
\SetKwComment{Comment}{$\triangleright$\ }{}
\SetKwInput{KwRequire}{Require}
\SetKwInput{KwEnsure}{Ensure}

\newcommand{\sysname}{\textsf{Amora}\xspace}
\newcommand{\vecpir}{\textsf{VecPIR}\xspace}
\newcommand{\nbrpir}{\textsf{NbrPIR}\xspace}
\newcommand{\rev}[1]{#1}
\begin{document}

\renewcommand{\thetheorem}{\arabic{theorem}}

\title{Amora: Communication-Efficient Private Approximate Nearest Neighbors Search}

\author{Qiuhao Wang}
\orcid{0009-0002-4381-5026}
\affiliation{%
  \institution{Xi'an Jiaotong University}
  \city{Xi'an}
  \country{China}
}
\email{QiuHaoWang@stu.xjtu.edu.cn}

\author{Zuowei Chen}
\orcid{0009-0003-4447-0979}
\affiliation{%
  \institution{Xi'an Jiaotong University}
  \city{Xi'an}
  \country{China}
}
\email{chenzuowei\_xjtu@163.com}

\author{Xu Yang}
\orcid{0000-0003-0094-6245}
\affiliation{%
  \institution{Xi'an Jiaotong University}
  \city{Xi'an}
  \country{China}
}
\email{yangxu@stu.xjtu.edu.cn}

\author{Saiyu Qi}
\authornote{Corresponding Author}
\orcid{0000-0002-0394-4432}
\affiliation{%
  \institution{Xi'an Jiaotong University}
  \city{Xi'an}
  \country{China}
}
\email{saiyu-qi@mail.xjtu.edu.cn}

\author{Yong Qi}
\orcid{0000-0002-7682-5653}
\affiliation{%
  \institution{Xi'an Jiaotong University}
  \city{Xi'an}
  \country{China}
}
\email{qiy@mail.xjtu.edu.cn}

\begin{abstract}
Approximate nearest neighbors search (ANNS) over cloud-hosted vector databases raises a critical privacy concern: the query vector, which often encodes sensitive user intent, is exposed to the server. Private ANNS (PANNS) allows a client to retrieve nearest neighbors without revealing its query. The state-of-the-art single-server scheme, Pacmann (ICLR'25), achieves high accuracy via graph-based search with private information retrieval (PIR), but incurs substantial communication overhead because each PIR query retrieves only a single vertex. We present \sysname, a communication-efficient PANNS system that applies SimplePIR to a partitioned graph index so that one PIR query is amortized over many vertices. To retrieve as much of the distance and neighbor information the search needs under a limited communication budget, \sysname adopts a \emph{dual-PIR} split layout with two heterogeneous PIR matrices that stay balanced in dimension: \vecpir retrieves inner products for an entire partition of vertices as their distance at once, while \nbrpir retrieves necessary neighbor lists at subgroup granularity for graph traversal. This layout shifts the cost model of PIR query from per-vertex to per-partition, which significantly reduces the communication overhead. To further improve efficiency by reducing the number of scans and query rounds in the PIR database, \sysname introduces \emph{batch PIR} to pack multiple queries into one matrix multiplication, \emph{centroid-based pruning} to filter unpromising partitions, and \emph{TopCand exploration} that seeds jump-style graph navigation with the closest vertices from each retrieved partition. We evaluate \sysname on \rev{five} datasets spanning 96 to \rev{1024} dimensions and up to 10 million vectors. Compared to Pacmann, \sysname reduces online communication by \rev{93--99\%} while achieving higher search accuracy.
\end{abstract}



\maketitle
\enlargethispage{-\baselineskip}  

\section{Introduction}
\label{sec:intro}
  
Vector similarity search has become a foundational operation in modern intelligent systems, underpinning recommendation engines~\cite{DBLP:conf/kdd/GrbovicC18, DBLP:journals/tbd/DouzeGDJSMLHJ26}, semantic search~\cite{DBLP:journals/corr/abs-1301-3781, DBLP:conf/emnlp/PenningtonSM14}, and retrieval-augmented generation (RAG)~\cite{DBLP:conf/nips/LewisPPPKGKLYR020, DBLP:conf/emnlp/KarpukhinOMLWEC20} for large language models. At its core lies the $k$-nearest neighbor search ($k$-NNS) problem: given a high-dimensional database and a query, find the $k$ vectors in the database closest to the query. Exact $k$-NNS is prohibitively expensive at scale, so practical systems rely on approximate nearest neighbors search (ANNS), which trades a small loss in accuracy for speedups via graph-based indices~\cite{DBLP:journals/is/MalkovPLK14, DBLP:journals/pami/MalkovY20, DBLP:journals/corr/abs-1810-07355, DBLP:journals/pvldb/FuXWC19, DBLP:conf/nips/SubramanyaDSKK19}, hash-based~\cite{DBLP:conf/stoc/IndykM98, DBLP:conf/vldb/GionisIM99}, or cluster-based partitioning~\cite{DBLP:journals/pami/JegouDS11, DBLP:conf/nips/ChenZWLLLYW21}. As these ANNS systems increasingly migrate to cloud-hosted vector databases~\cite{DBLP:journals/tbd/JohnsonDJ21, DBLP:conf/sigmod/WangYGJXLWGLXYY21}, a critical privacy concern arises: the search query, which often encodes sensitive user intent such as medical symptoms, legal questions, or personal preferences, is transmitted to the server, exposing what the client is searching for. Private ANNS (PANNS) addresses this by allowing a client to retrieve the $k$ nearest neighbors of a query vector without revealing the query to the server.

One line of work builds PANNS under a multi-server setting~\cite{DBLP:conf/sp/Servan-Schreiber22, DBLP:conf/esorics/SongLCPW23}, where the database is replicated across servers assumed not to collude. However, since the non-collusion assumption is difficult to enforce in practice, clients may trust such schemes less than single-server ones. Recent single-server schemes avoid this assumption and fall into two families: cluster-based and graph-based schemes. Cluster-based schemes SANNS~\cite{DBLP:conf/uss/0030CDPRR20}, Panther~\cite{DBLP:conf/ccs/LiHZ0L0C25}, and Tiptoe~\cite{DBLP:conf/sosp/HenzingerDCZ23} group similar vectors into clusters and use oblivious RAM (ORAM)~\cite{DBLP:journals/jacm/GoldreichO96, DBLP:journals/jacm/StefanovDSCFRYD18} or private information retrieval (PIR)~\cite{DBLP:journals/pacmmod/ZhangPSLLWLGZ25, DBLP:conf/crypto/LiMRS24, DBLP:conf/icde/DongWWZC25} to privately search the clusters whose centroids are closest to the query. \rev{However, a centroid summarizes an entire cluster, and thus can only coarsely guide nearest neighbor searches. A plaintext ANNS system compensates by searching more clusters at little extra cost, but in the private setting each additional cluster costs one expensive PIR query. As a result, coarse guidance and a limited communication budget together leave some true neighbors in unvisited clusters, fundamentally capping accuracy.}

Graph-based schemes Pacmann~\cite{DBLP:conf/iclr/ZhouSF25} and Compass~\cite{DBLP:conf/osdi/ZhuPZP25} instead organize the vectors into a proximity graph such as HNSW~\cite{DBLP:journals/pami/MalkovY20}, in which each vertex stores one vector and its neighbor list. The search maintains a candidate set and expands the closest candidate with PIR or ORAM at each hop to protect the query privacy. \rev{The edges provide finer and more precise guidance than centroid distances, so reaching the same accuracy requires visiting far fewer vertices.} Pacmann is the state of the art for query privacy over a public database and achieves high accuracy, but it incurs substantial communication overhead in both search and preprocessing, making it impractical for bandwidth-constrained clients. This stems from Pacmann's underlying PIR scheme, Piano~\cite{DBLP:conf/sp/ZhouPZS24}, whose preprocessing requires downloading the entire database. More critically, Pacmann exhibits low information density in its online search: each PIR query incurs $O(\sqrt{N})$ communication yet retrieves only a single vertex's vector and neighbor list. We refer to this as the \emph{per-vertex} cost model, and it is the fundamental source of Pacmann's inefficiency.

We observe that SimplePIR~\cite{DBLP:conf/uss/HenzingerHCMV23} offers an inherent information-density advantage. It organizes the database as a matrix in which one query retrieves an entire column, and attains the optimal communication on a balanced matrix. Thus, one target vertex costs the same $O(\sqrt{N})$ communication, and every other vertex in that column comes for free. \rev{A direct application of SimplePIR partitions the graph and stores each partition's vertices in one column, so a single PIR query returns a group of closely connected vertices. However, this naive method fails to realize the advantage. Since each vertex carries its neighbor list together with its vector for graph navigation, the two payloads make a vertex heavy and fill many entries of a single column. Such heavy payloads create a dilemma in designing the PIR matrix layout. One option is to adopt a balanced PIR matrix for the optimal communication cost. However, the heavy payloads then allow only a small number of vertices per column, so each query returns little information, and the search needs far more queries and rounds to reach high accuracy. The other option stores more vertices per column to raise the information per PIR query, but the matrix then becomes unbalanced, and every query costs more.} The first challenge is therefore: \textbf{C1.}~how to organize the PIR database so that a single PIR query carries as much useful vertex information as possible for the nearest neighbor search while keeping the optimal communication cost. Moreover, graph-based ANNS requires multi-hop traversal, and each hop invokes PIR to protect query privacy. Although SimplePIR is built on lightweight Learning with Errors (LWE)~\cite{DBLP:journals/jacm/Regev09}, each query must scan the entire database, so it incurs $O(N)$ computation and makes memory bandwidth a bottleneck. The second challenge is therefore: \textbf{C2.}~how to reduce the number of scans and query rounds in the PIR database to improve search efficiency. 

To address both challenges, we build \sysname, a communication-efficient PANNS scheme with high accuracy. \rev{For \textbf{C1}, our insight is that a vertex's two payloads serve the search differently: its distance can improve the candidate set, whereas its neighbor list matters only if the search expands it. \sysname therefore adopts a \emph{dual-PIR} split layout that organizes the two types of payloads into two separate PIR databases with vertices grouped by graph partition: (1) \vecpir spreads each vertex's vector horizontally across columns, so that each partition holds far more vectors and one PIR query computes inner products with an entire partition to return its distance to every vertex in it; (2) \nbrpir stacks neighbor lists vertically and cuts each partition into subgroups to return the neighbor lists of only the candidates that the search expands. As a result, both PIR matrices stay balanced, and the search alternates between them. A round of PIR queries now amortizes its communication over a whole partition, which shifts the cost model from per-vertex to per-partition.} As for \textbf{C2}, we introduce three optimization techniques that reduce the PIR queries needed to reach a target accuracy: (1) \emph{Batch PIR} merges concurrent PIR queries into a single matrix multiplication, reducing several communication rounds to one and amortizing SimplePIR's linear database scan across queries; (2) \emph{Centroid-based pruning} ranks pending partitions by the plaintext distance from the query to their centroids and processes only the most promising ones, avoiding PIR queries to distant regions; (3) \emph{TopCand exploration} extracts the closest vertices (top candidates) from each retrieved partition as seeds in the candidate set, enabling jump-style navigation into distant regions rather than strictly following edges. 

In general, our contributions are as follows:

\begin{itemize}[leftmargin=*,nosep]
    \item We design a dual-PIR split layout that combines SimplePIR with graph partitioning, storing vectors horizontally and neighbor lists in subgroup columns. Both PIR matrices stay balanced, so every PIR query keeps SimplePIR's optimal communication, while a round of PIR queries amortizes its cost over a whole partition, shifting the cost model from per-vertex to per-partition.
    \item We propose three optimization techniques to reduce rounds and database scans: batch PIR merges concurrent queries into a single matrix multiplication, centroid-based pruning ignores partitions with distant centroids, and \emph{TopCand} exploration extracts seeds from each retrieved partition for jump-style graph navigation.
    \item We open-source \sysname\footnote{\url{https://github.com/meander-wqh/Amora/}.}, prove query privacy, analyze communication and computation complexity, and evaluate it on \rev{five} datasets spanning 96--\rev{1024} dimensions and up to 10M vectors. Compared to Pacmann, \sysname reduces online communication by \rev{93--99\%} while achieving higher search accuracy.
\end{itemize}

\section{Related Work}
\label{sec:related}
\subsection{Approximate Nearest Neighbors Search}

Approximate nearest neighbors search (ANNS) algorithms can be broadly divided into cluster-based, hash-based and graph-based categories. Cluster-based methods such as IVF~\cite{DBLP:journals/pami/JegouDS11}, ScaNN~\cite{DBLP:conf/icml/GuoSLGSCK20}, and SPANN~\cite{DBLP:conf/nips/ChenZWLLLYW21} partition the vector space and restrict the search to the most promising partitions. Their fixed access pattern is easier to protect with PIR, but restricting search to a few clusters limits accuracy when the true neighbors span multiple partitions. Hash-based methods~\cite{DBLP:conf/stoc/IndykM98, DBLP:conf/vldb/GionisIM99} such as Locality Sensitive Hashing (LSH)~\cite{DBLP:conf/compgeom/DatarIIM04} map vectors into hash buckets so that nearby vectors are likely to collide, offering sublinear query time with theoretical guarantees but requiring many hash tables to achieve high recall in practice. Graph-based methods build a navigable proximity graph and answer queries via greedy traversal, exploiting the small-world property to reach nearest neighbors in logarithmic hops; representative schemes include NSW~\cite{DBLP:journals/is/MalkovPLK14}, HNSW~\cite{DBLP:journals/pami/MalkovY20}, NGT~\cite{DBLP:journals/corr/abs-1810-07355}, NSG~\cite{DBLP:journals/pvldb/FuXWC19}, and DiskANN~\cite{DBLP:conf/nips/SubramanyaDSKK19}. Graph-based methods currently achieve the best accuracy--throughput trade-off on most benchmarks, but their iterative, data-dependent traversal pattern makes privacy protection challenging, as each hop reveals which vertex is being accessed.

\subsection{Private Approximate Nearest Neighbors Search}

Some PANNS works adopt a \emph{multi-server} setting. Servan-Schreiber et al.~\cite{DBLP:conf/sp/Servan-Schreiber22} combine LSH with distributed point functions~\cite{DBLP:conf/ccs/BoyleGI16} across two non-colluding servers, and Song et al.~\cite{DBLP:conf/esorics/SongLCPW23} further reduce the required subsets. However, the non-collusion assumption is difficult to enforce in practice, motivating \emph{single-server} approaches.

Among single-server schemes, cluster-based methods appear first. SANNS~\cite{DBLP:conf/uss/0030CDPRR20} combines ORAM~\cite{DBLP:journals/jacm/GoldreichO96} and garbled circuits~\cite{DBLP:conf/focs/Yao86} but incurs several gigabytes of communication per query. Panther~\cite{DBLP:conf/ccs/LiHZ0L0C25} replaces ORAM with a co-design of batch PIR~\cite{DBLP:conf/sp/AngelCLS18}, secret sharing~\cite{DBLP:journals/cacm/Shamir79}, and homomorphic encryption~\cite{DBLP:conf/stoc/Gentry09}, reducing communication while protecting both query and database privacy. Tiptoe~\cite{DBLP:conf/sosp/HenzingerDCZ23} focuses on query privacy only and uses somewhat homomorphic encryption~\cite{DBLP:journals/toct/BrakerskiGV14} to privately compute distances within $k$-means clusters in a single round. Wally~\cite{DBLP:journals/corr/abs-2406-06761} improves Tiptoe by batching anonymous queries and relaxing privacy to differential privacy~\cite{DBLP:conf/tcc/DworkMNS06}. However, all cluster-based schemes restrict the search to a small number of partitions and lack iterative refinement, fundamentally limiting their accuracy.

To overcome this accuracy barrier, graph-based schemes leverage multi-hop traversal. Pacmann~\cite{DBLP:conf/iclr/ZhouSF25} runs graph-based search using Piano PIR~\cite{DBLP:conf/sp/ZhouPZS24} to privately retrieve per-vertex information at each hop, achieving significantly higher accuracy than cluster-based methods. Compass~\cite{DBLP:conf/osdi/ZhuPZP25} combines HNSW with Ring ORAM~\cite{DBLP:journals/jacm/StefanovDSCFRYD18, DBLP:conf/uss/RenFKSSDD15} under a setting where the client owns the encrypted database, matching plaintext accuracy. However, both of them require significant client-side storage and communication overhead. PrivANN~\cite{DBLP:conf/trustcom/ChenCLSL25} uses a combination of ORAM and a Trusted Execution Environment (TEE)~\cite{DBLP:journals/iacr/CostanD16} to provide high-accuracy and low communication PANNS, but its security assumptions rely on a TEE, which introduces additional side-channel attacks~\cite{DBLP:conf/dac/MishraABSM25, DBLP:conf/acsac/Lang0MLWL22, DBLP:conf/date/KouSHZ23}.

\section{Background}
\label{sec:background}

\subsection{Private Approximate Nearest Neighbors Search}

\subsubsection{Approximate Nearest Neighbors Search}
Given a $d$-dimensional metric space with a distance function $f_\Delta(\cdot,\cdot)$ and a database $\mathsf{DB} = \{\vec{v}_1, \vec{v}_2, \ldots, \vec{v}_N\} \subset \mathbb{R}^d$ containing $N$ vectors, the ANNS problem is: given a query vector $\vec{v}_q \in \mathbb{R}^d$, find an index set $I = \{i_1, \ldots, i_k\}$ such that the corresponding vectors $\{\vec{v}_{i_1}, \ldots, \vec{v}_{i_k}\}$ approximate the true $k$ nearest neighbors of $\vec{v}_q$ in $\mathsf{DB}$ under $f_\Delta$. The quality of the approximation is measured by Recall (the fraction of true nearest neighbors returned) or mean reciprocal rank (MRR).

\subsubsection{Graph-Based ANNS}
Graph-based ANNS builds a proximity graph $G = (V, E)$ over the database, where each vertex $v \in V$ corresponds to a database vector and each directed edge $(v, u) \in E$ connects $v$ to one of its approximate nearest neighbors. Every vertex maintains a neighbor list of at most $M_0$ outgoing edges. Given a query $\vec{q}$ (ANNS for vector $\vec{v}_q$), search proceeds by greedy traversal: starting from an entry vertex, the algorithm maintains a candidate heap $\mathcal{W}$ (a min-heap ordered by distance to $\vec{v}_q$) and a result set $\mathcal{R}$ (max-heap of capacity $\mathit{ef} \geq k$). At each step, the closest unvisited vertex is popped from $\mathcal{W}$, its neighbors are retrieved, their distances to $\vec{v}_q$ are computed, and those closer than the current worst result in $\mathcal{R}$ are inserted into both $\mathcal{W}$ and $\mathcal{R}$. The search terminates when $\mathcal{W}$ is empty, and the top-$k$ entries of $\mathcal{R}$ are returned. 

\subsubsection{Single-Server PANNS}
In many applications the query vector encodes sensitive user intent and should not be revealed to the server hosting the database. The privacy requirement for a PANNS protocol focused on in this paper is query privacy only, the server learns negligible information about the query $\vec{q}$, and the database is public and does not require protection. Following prior work~\cite{DBLP:conf/iclr/ZhouSF25, DBLP:conf/sosp/HenzingerDCZ23}, we consider the \emph{single-server preprocessing} model, in which the private protocol proceeds in two phases:

\begin{enumerate}[leftmargin=*,nosep]
    \item \textbf{Preprocessing.} The server holds the public database $\mathsf{DB} \in \mathbb{R}^{d \times N}$. The client and server interact to generate a compact hint that the client stores locally for use during subsequent queries.
    \item \textbf{Query.} The client issues a query $\vec{q}$. The client and server engage in one or more rounds of communication, after which the client outputs the index set $I$ of $k$-approximate nearest neighbors.
\end{enumerate}

\subsection{Private Information Retrieval}

\subsubsection{LWE and Regev Encryption}
SimplePIR's security relies on the Learning with Errors (LWE) assumption~\cite{DBLP:journals/jacm/Regev09}. Parameterized by a secret dimension $n$, a modulus $q \geq 2$, and an error distribution $\chi$ over $\mathbb{Z}$, LWE asserts that for a random matrix $A \stackrel{\$}{\leftarrow} \mathbb{Z}_q^{m \times n}$, a secret $s \stackrel{\$}{\leftarrow} \mathbb{Z}_q^{n}$, and an error vector $e \stackrel{\$}{\leftarrow} \chi^{m}$, the pair $(A,\; As + e)$ is computationally indistinguishable from $(A,\; r)$ where $r \stackrel{\$}{\leftarrow} \mathbb{Z}_q^{m}$. Based on LWE, Regev~\cite{DBLP:journals/jacm/Regev09} constructs a secret-key encryption scheme with an additional plaintext modulus $p$. The secret key is $s \stackrel{\$}{\leftarrow} \mathbb{Z}_q^{n}$, and the encryption of a message $\mu \in \mathbb{Z}_p$ is:
$(a,\; c) = (a,\; a^\mathsf{T} s + e + \lfloor q/p \rfloor \cdot \mu) \in \mathbb{Z}_q^{n} \times \mathbb{Z}_q,$
where $a \stackrel{\$}{\leftarrow} \mathbb{Z}_q^{n}$ and $e \stackrel{\$}{\leftarrow} \chi$. Decryption computes $c - a^\mathsf{T} s \bmod q$ and rounds to the nearest multiple of $\lfloor q/p \rfloor$. A key property is additive homomorphism: the component-wise sum of two ciphertexts decrypts to the sum of the plaintexts, provided the accumulated error remains small.

\subsubsection{SimplePIR}
Private information retrieval (PIR)~\cite{DBLP:journals/jacm/ChorKGS98, DBLP:conf/focs/KushilevitzO97, DBLP:conf/ndss/SongXZZLMX26, DBLP:conf/uss/HenzingerHCMV23, DBLP:journals/pacmmod/ZhangPSLLWLGZ25, DBLP:conf/uss/Menon024, DBLP:conf/crypto/LiMRS24, DBLP:conf/icde/DongWWZC25} allows a client to fetch record $i$ from a server's database of $N$ records without revealing $i$ to the server. SimplePIR~\cite{DBLP:conf/uss/HenzingerHCMV23} is a single-server PIR scheme that achieves high throughput by leveraging Regev encryption and server-side preprocessing. SimplePIR organizes the $N$-element database as a $\sqrt{N} \times \sqrt{N}$ balanced matrix $\mathsf{db} \in \mathbb{Z}_p^{\sqrt{N} \times \sqrt{N}}$, and defines the scalar $\Delta := \lfloor q/p \rfloor$. Its protocol consists of four interfaces:

\begin{itemize}[leftmargin=*,nosep]
    \item $\mathsf{hint}_c \leftarrow \mathsf{Setup}(\mathsf{db})$: The server computes $\mathsf{hint}_c \leftarrow \mathsf{db} \cdot A$ where $A \in \mathbb{Z}_q^{\sqrt{N} \times n}$ is a shared LWE matrix. This is a one-time offline phase; $\mathsf{hint}_c$ is sent to the client and reused across all queries.
    \item $(\mathsf{st}, qu) \leftarrow \mathsf{Query}(i)$: The client first decomposes index $i$ into $(i_\text{row}, i_\text{col})$, samples $s \stackrel{\$}{\leftarrow} \mathbb{Z}_q^n$ and $e \stackrel{\$}{\leftarrow} \chi^{\sqrt{N}}$, and computes $qu \leftarrow As + e + \Delta \cdot u_{i_\text{col}} \in \mathbb{Z}_q^{\sqrt{N}}$, where $u_{i_\text{col}}$ is the unit vector with a ``1'' at index $i_\text{col}$. The client stores $\mathsf{st} = (i_\text{row}, s)$ and sends $qu$ to the server.
    \item $ans \leftarrow \mathsf{Answer}(\mathsf{db}, qu)$: The server computes $ans \leftarrow \mathsf{db} \cdot qu \in \mathbb{Z}_q^{\sqrt{N}}$ and returns it to the client.
    \item $d \leftarrow \mathsf{Recover}(\mathsf{st}, \mathsf{hint}_c, ans)$: The client first computes $\hat{d} \leftarrow ans[i_\text{row}] - \mathsf{hint}_c[i_\text{row},:] \cdot s \in \mathbb{Z}_q$, and it recovers the plaintext $d = \mathrm{Round}_\Delta(\hat{d})/\Delta \in \mathbb{Z}_p$.
\end{itemize}

By the additive homomorphism of Regev encryption, $ans$ is an element-wise encryption of the $i_\text{col}$-th column of $\mathsf{db}$. Therefore, the client can run $\mathsf{Recover}$ for every row index to reconstruct the entire column from a single query.


\subsection{Threat Model}
\label{sec:threat}

\sysname targets the same setting as Pacmann~\cite{DBLP:conf/iclr/ZhouSF25}: a single server $\mathcal{S}$ holds a public vector database $\mathsf{DB} \in \mathbb{R}^{d \times N}$ of $N$ vectors each with $d$ dimensions, and a stateful client $\mathcal{C}$ wishes to perform ANNS without revealing the query vector to the server. The database is public and does not require protection, i.e., we focus exclusively on \emph{query privacy}. We assume the adversary is \emph{semi-honest} (honest-but-curious). That is, the adversary follows the protocol specification but may attempt to learn additional information about the query from leakages. The privacy of the protocol requires that the server learns negligible information about the query $\vec{q}$. Formally, we use a simulation-based definition for the privacy of the protocol:

\begin{definition}[Query Privacy]
\label{def:privacy}
An ANNS protocol is private if there exists a simulator $\mathsf{Sim}$ such that for any probabilistic polynomial-time (PPT) adversary $\mathcal{A}$ acting as the server, the views of $\mathcal{A}$ in the following two experiments are computationally indistinguishable w.r.t.\ the security parameter \rev{$\lambda$}\footnote{\rev{$\lambda$ denotes the bit-security of the scheme and is conventionally set to 128: breaking the underlying LWE problem then requires on the order of $2^{128}$ operations, the standard security level of modern cryptographic systems.}}:
\begin{itemize}[leftmargin=*,nosep]
    \item \textbf{Real:} The client interacts with $\mathcal{A}(1^\lambda, \mathsf{DB})$ who acts as the server. The client may issue multiple ANNS queries sequentially and each ANNS query $\vec{q}$ takes a series of PIR searches as input and involves multiple rounds of communication between the client and the server. After observing the transcript of previous queries, $\mathcal{A}$ may adaptively choose the next query $\vec{q}$ for the client.
    \item \textbf{Ideal:} The simulated client $\mathsf{Sim}(1^\lambda, N)$ interacts with $\mathcal{A}(1^\lambda, \mathsf{DB})$. $\mathcal{A}$ still may adaptively choose the next query $\vec{q}$. However, the simulator is invoked with only the knowledge of the database size $N$, and without the information of the chosen query $\vec{q}$.
\end{itemize}
\end{definition}

\rev{We note that client-side timing side channels are out of scope: a server can observe how long the client takes to respond between rounds. If needed, they can be removed by padding the client's per-round work to a fixed budget.}

\begin{figure}[t]
\centering
\includegraphics[width=0.6\columnwidth]{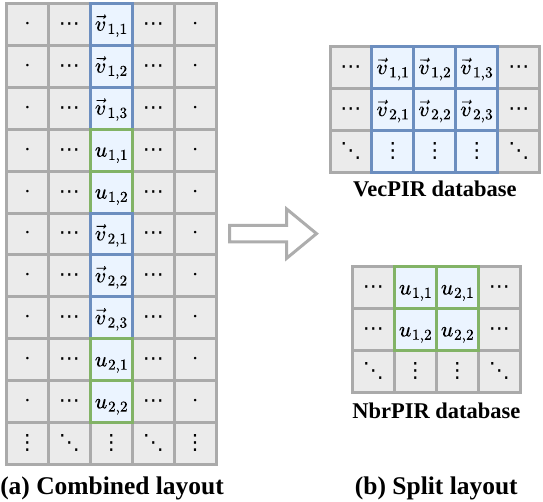}
\caption{\rev{Combined layout vs.\ dual-PIR split layout ($d{=}3$ dimensions, $M_0{=}2$ neighbors). $\vec{v}_{i,j}$ and $u_{i,j}$ are the $j$-th dimension and $j$-th neighbor of vertex $i$.}}
\label{fig:split}
\end{figure}
\section{Design Overview}
\label{sec:design}

\sysname follows the graph-based paradigm: the client keeps a candidate heap and, at each of a fixed number of hops, picks the vertex closest to the query and retrieves its neighbor list from the server via PIR. Each vertex therefore carries two payloads: its vector for distance computation and its neighbor list for traversal. To amortize one PIR query over many vertices, the direct design applies SimplePIR to a graph partitioned by METIS~\cite{DBLP:journals/siamsc/KarypisK98} into $C$ partitions of at most $K$ vertices each, and stores each partition's vertices in a single column of the PIR matrix, forming the combined layout in Figure~\ref{fig:split}(a). Partitioning co-locates closely connected vertices, so one PIR query serves an entire partition. \rev{Moreover, it is only a logical assignment of vertices and removes no edge, so each vertex keeps its full neighbor list and the search can still reach other partitions.} Since a vertex carries both payloads, it fills many entries of a column, which creates a dilemma: (1) Keeping the PIR matrix balanced for the optimal communication cost limits each column to a few vertices, so one PIR query amortizes its cost over too few of them; (2) Packing more vertices per column instead makes the matrix unbalanced, raising the communication of every PIR query.

\begin{figure}[t]
\centering
\includegraphics[width=0.65\columnwidth]{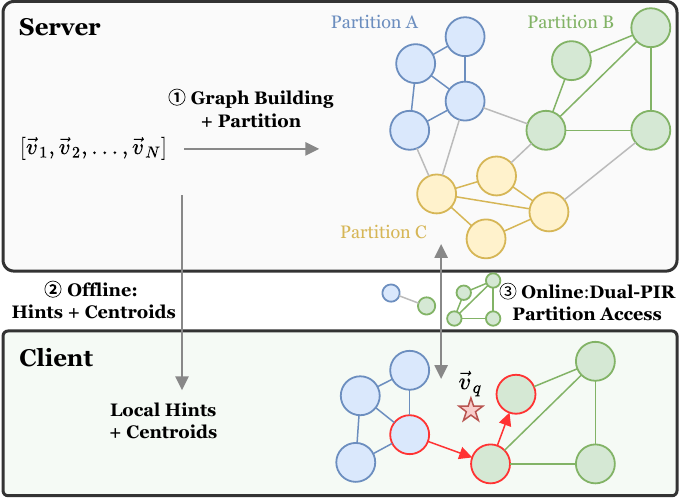}
\caption{High-level overview of \sysname. \ding{172}~The server builds a graph index and partitions it. \ding{173}~Offline: the server sends PIR hints and centroids to the client. \ding{174}~Online: the client issues \vecpir queries for partition-level distances and \nbrpir queries for subgroup-level neighbor lists.}
\label{fig:overview}
\end{figure}

The first challenge \textbf{C1} is therefore: how to organize the PIR database so that a single PIR query carries as much useful vertex information as possible for the nearest neighbor search while keeping the optimal communication cost. \rev{Our key insight is that the search needs the two payloads in different amounts: every retrieved vertex contributes its distance, whereas only the few vertices that the search expands need their neighbor lists. The combined layout ignores this asymmetry by coupling one neighbor list to every vector, so each PIR query spends part of its communication budget on unnecessary neighbor lists. The two payloads therefore call for different arrangements, each retrieved at the granularity the search actually needs. \sysname addresses \textbf{C1} by splitting the two payloads into a \emph{dual-PIR} split layout, as Figure~\ref{fig:split}(b) shows}. The \vecpir matrix stores only vector information: each partition occupies a block of $d$ consecutive columns, and within that block every vertex takes one row, with its $d$-dimensional vector spread horizontally across the $d$ columns. Each column then stores only one dimension per vertex, so a balanced matrix holds far more vectors per partition. One query computes inner products with all $K$ vertices at once, giving their distances to the query. The client then selects candidates from these distances and uses \nbrpir to retrieve only the neighbor lists it needs to expand the candidates. Storing an entire partition's neighbor lists in one column, that is, $K$ vertices with $M_0$ neighbors each, would both exceed this need and make the \nbrpir matrix highly unbalanced. We therefore organize the \nbrpir matrix into fine-grained subgroups of $S$ vertices within each partition, so that each subgroup occupies one column, keeping the matrix balanced. 


The dual-PIR layout improves the information density of PIR queries. The remaining cost is the number of scans and query rounds in the PIR database that the traversal needs, and how to reduce them is the second challenge \textbf{C2}. \sysname introduces three optimization techniques: (1) \emph{Batch PIR} packs multiple independent PIR queries into a single matrix multiplication, reducing multiple communication rounds to one and alleviating SimplePIR's memory-bandwidth bottleneck; (2) \emph{Centroid-based pruning} ranks pending partitions by the plaintext distance from the query to their centroids and discards the distant ones, which incurs zero privacy cost since centroids are public metadata; (3) \emph{TopCand exploration} extracts the $t_c$ closest vertices from each retrieved partition and injects them as seeds into the candidate heap, enabling jump-style graph navigation that leaps into distant regions rather than strictly following edges.

The system operates in two phases. In the offline phase (executed once), the server constructs the graph index, applies graph partitioning, decomposes each partition into subgroups, and builds the \vecpir and \nbrpir databases. The server also computes the partition centroids which will be sent to the client. The server then runs $\mathsf{Setup}$ for each PIR module, computing $\mathsf{hint}$ matrices that the client downloads and stores locally. In the online phase (per query), the client first performs centroid entry: it computes $\vec{v}_q$'s distance to all $C$ partition centroids locally and selects the $K_e$ nearest partitions as entry points, requiring zero communication. \rev{The search then proceeds for $T$ iterations, each consisting of a batched \nbrpir query followed by a batched \vecpir query.} Because the number of iterations and the batch sizes are determined by public system parameters, the communication pattern is identical for every query and leaks no information about the query content. Figure~\ref{fig:overview} illustrates this two-phase workflow.

\begin{table}[t]
\centering
\caption{\rev{Notation summary.}}
\label{tab:notation}
\small
\setlength{\tabcolsep}{5pt}
\rev{\begin{tabular}{@{}l p{0.74\columnwidth}@{}}
\toprule
Notation & Meaning \\
\midrule
$N$ & the number of database vectors \\
$d$ & the vector dimension \\
$M_0$ & the maximum number of neighbors per vertex \\
$C$ & the number of partitions \\
$K, S$ & the maximum sizes of a partition and a subgroup \\
$g_{\max}$ & the maximum number of subgroups per partition \\
$t$ & the number of parts in split encoding \\
$n$ & the LWE dimension \\
$\log q$ & the LWE modulus \\
$\lambda$ & the LWE security level ($\approx$ 128-bit) \\
$k$ & the number of returned nearest neighbors \\
$\mathit{ef}$ & the result-heap capacity ($\ge k$) \\
$B_v, B_n$ & the \vecpir and \nbrpir batch sizes (the numbers of partitions and subgroups queried per iteration) \\
$T$ & the number of search iterations (each issues one batched \nbrpir and one batched \vecpir query) \\
$K_e$ & the number of entry partitions \\
$t_c$ & the number of TopCand seeds taken from each partition \\
$\alpha$ & the pruning ratio \\
\bottomrule
\end{tabular}}
\end{table}

\section{System Design}
\label{sec:system}

\rev{In this section, we elaborate on the detailed system design of \sysname, and summarize the notation used throughout the paper in Table~\ref{tab:notation}.}
\subsection{\vecpir: Distance Computation}
\label{sec:vecpir}

In SimplePIR, querying column $i$ returns all elements in that column. If each vertex's vector is stored as a single column element, one query would retrieve only a few vertices' vectors. The key design of \vecpir is horizontal arrangement (Figure~\ref{fig:dual-pir} (a)): each vertex's $d$-dimensional quantized vector is spread across $d$ consecutive columns, with vertices in the same partition stacked as rows. Concretely, the \vecpir database $\mathsf{db}_{\text{vec}}$ is a $K \times (C \cdot d)$ matrix: partition $c$ occupies column block $[c \cdot d,\; (c+1) \cdot d)$, and entry $(x,\; c \cdot d + j)$ stores the $j$-th dimension of the $x$-th vertex's quantized vector in partition $c$. This layout transforms a column retrieval into an inner-product computation. Let $\vec{v}_q \in \mathbb{Z}_p^d$ be the quantized query vector, $A \in \mathbb{Z}_q^{M \times n}$ the shared random matrix with $M = C \cdot d$, $s \in \mathbb{Z}_q^n$ a fresh secret, and $e \in \chi^M$ a noise vector. To query partition $c$, the client constructs:
\begin{equation}
\label{eq:embquery}
qu = As + e + \Delta \cdot q_c \in \mathbb{Z}_q^M
\end{equation}
where $q_c \in \mathbb{Z}^M$ places $\vec{v}_q$ at positions $[c \cdot d,\; (c+1) \cdot d)$ and is zero elsewhere, and $\Delta = \lfloor q / p \rfloor$ is the scaling factor. The server computes $\mathsf{ans} = \mathsf{db}_{\text{vec}} \cdot qu \in \mathbb{Z}_q^{K}$ and returns it to the client.

For row $i$, the answer contains $\langle \vec{v}_q,\; \mathsf{db}_{\text{vec}}[i, c \cdot d : (c+1) \cdot d] \rangle + \text{noise}$, the noisy inner product between the query and the $i$-th vertex. The client recovers the inner product:
\begin{equation}
\hat{d}_i = \left\lfloor \frac{\mathsf{ans}[i] - \mathsf{hint}_v[i] \cdot s}{\Delta} \right\rceil
\end{equation}
where $\mathsf{hint}_v = \mathsf{db}_{\text{vec}} \cdot A$ is the precomputed hint matrix, and $\lfloor \cdot \rceil$ denotes rounding to the nearest integer. A single query thus yields inner products with all $K$ vertices in the target partition. The plaintext modulus $p$ is chosen to retain sufficient precision for ranking.

\subsection{\nbrpir: Neighbor Retrieval}
\label{sec:nbrpir}

\subsubsection{Split encoding}
Each neighbor is identified by its partition index $c \in \{0,\ldots,C-1\}$ and its local index $\ell \in \{0,\ldots,K-1\}$ within that partition. Let $b_\ell = \lceil \log_2 K \rceil$ and $b_c = \lceil \log_2 C \rceil$ denote the number of bits required for each field. The two indices are combined into a single non-negative integer $id = c \cdot 2^{b_\ell} + \ell$, which occupies $b_c + b_\ell$ bits in total. Because a single SimplePIR entry can hold at most $b$ bits (the per-entry bit width), the integer $id$ is split into $t = \lceil (b_c + b_\ell) / b \rceil$ parts: denoting the $i$-th part ($i = 0,\ldots,t-1$) as $id_i$, each part is stored in a separate matrix row and recovered independently. The plaintext modulus for \nbrpir is $p_n = 2^b$, far smaller than \vecpir's, so its scaling factor is correspondingly larger and \nbrpir has ample noise margin. 

\subsubsection{Subgroup organization}
Each vertex has at most $M_0$ neighbors. If an entire partition's neighbor lists are stored in one column, the column height would be $K \cdot M_0 \cdot t$ while the number of columns would be only $C$, producing an extremely unbalanced matrix that prevents SimplePIR from balancing upload and download costs. We resolve this by subdividing each partition into subgroups of at most $S$ vertices (Figure~\ref{fig:dual-pir} (b)). Each subgroup occupies one column, yielding at most $g_{\max} = \lceil K / S \rceil$ subgroup columns per partition. In this case, the \nbrpir database $\mathsf{db}_{\mathsf{nbr}}$ has $S \cdot M_0 \cdot t$ rows and $C \cdot g_{\max}$ columns, a much more balanced shape. Vertices with fewer than $M_0$ neighbors and subgroups with fewer than $S$ vertices are padded to maximums so that every column has exactly $S \cdot M_0 \cdot t$ entries.

\begin{figure*}[t]
\centering
\includegraphics[width=0.9\textwidth]{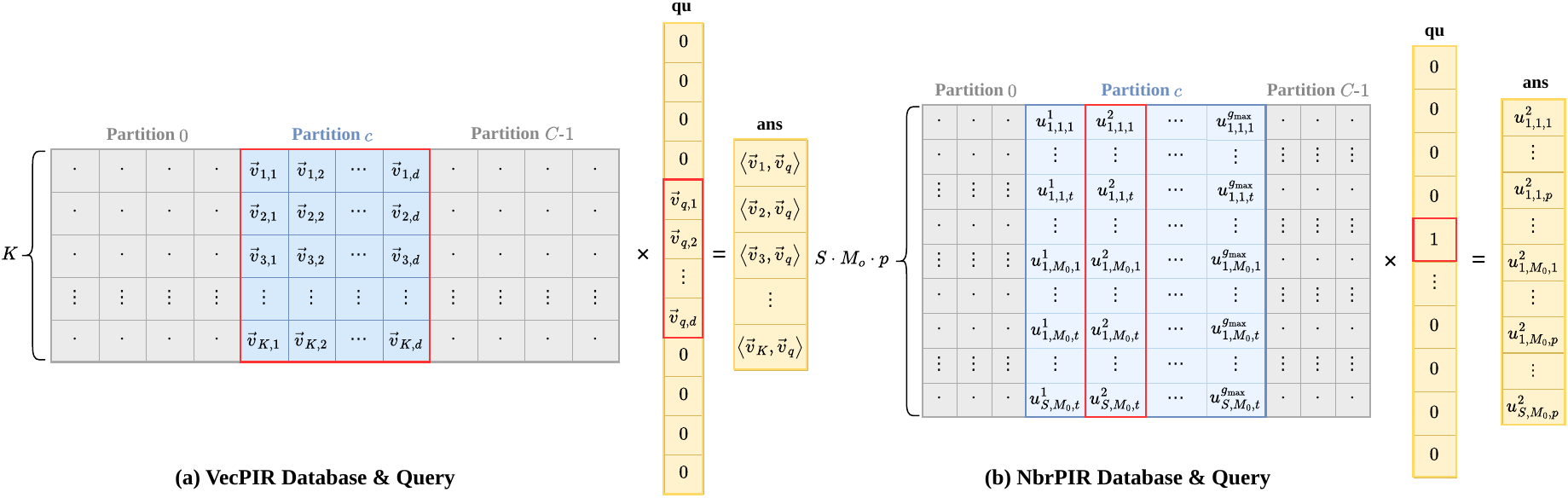}
\caption{Dual-PIR split layout. (a): \vecpir arranges each vertex's $d$-dimensional vector horizontally across $d$ columns within a partition block, querying partition $c$ returns inner products with all $K$ vertices. (b): \nbrpir organizes each partition into $g_{\max}$ subgroup columns, each subgroup stores $S$ vertices' $M_0$ neighbors using split encoding with $t$ parts.}
\label{fig:dual-pir}
\end{figure*}


\subsubsection{Vertex-to-subgroup mapping}
When \vecpir returns distances for partition~$c$, the client selects promising candidates by ranking the $K$ returned inner products. For a candidate at row~$\ell$ (i.e., local index~$\ell$ within partition~$c$), its neighbor list resides in \nbrpir subgroup column
$c \cdot g_{\max} + \lfloor \ell / S \rfloor$,
and its entries start at offset $(\ell \bmod S) \cdot M_0 \cdot t$ within that column. The client issues a SimplePIR column query for this subgroup and extracts the target vertex’s $M_0 \cdot t$ rows to recover its full neighbor list.

\subsection{Efficient Private Search}
\label{sec:opt}

The dual-PIR split layout amortizes the communication of a private search: each \vecpir query returns distances to all $K$ vertices in a partition, and each \nbrpir query retrieves neighbor lists for all $S$ vertices in a subgroup, so one round trip yields far more information than in per-vertex PIR systems. The search proceeds iteratively: pop the closest candidate from the heap, issue one \nbrpir query for its subgroup to decode its neighbors, then issue one \vecpir query for each neighbor's uncached partition to obtain distances. We further improve this basic search with the following techniques.

\subsubsection{Batch PIR}
The PIR database matrices $\mathsf{db}_{\text{vec}}$ and $\mathsf{db}_{\text{nbr}}$ can occupy multiple gigabytes of server memory, and each PIR query requires a full matrix-vector multiplication that scans the entire database. When $B$ PIR queries are issued independently, the server must scan the database $B$ times, making memory bandwidth the dominant performance bottleneck. Batch PIR addresses this by packing multiple queries into a single matrix multiplication. For \vecpir, to query $B_v$ partitions $\{c_1, \ldots, c_{B_v}\}$ simultaneously, the client constructs a $(C \cdot d) \times B_v$ query matrix $Q$ where each column is one partition's query vector as in Eq.~\eqref{eq:embquery}. The server computes $\mathsf{ans} = \mathsf{db}_{\text{vec}} \cdot Q \in \mathbb{Z}_q^{K \times B_v}$ with a single database scan, amortizing the memory-bandwidth cost across all $B_v$ queries while merging $B_v$ round trips into one and resolving distances to $K \cdot B_v$ vertices. The same batching applies to \nbrpir: a $(C \cdot g_{\max}) \times B_n$ query matrix yields an $(S \cdot M_0 \cdot t) \times B_n$ answer matrix containing $B_n$ subgroups' neighbor encodings in one round trip.

\subsubsection{TopCand exploration}
Standard graph search navigates solely by following edges, visiting one vertex at a time. \rev{In \sysname, to protect the query privacy, each \vecpir query privately selects one partition, and the server computes the distances to all vertices inside it within a single encrypted operation. The distances to all $K$ vertices in the partition are thus a byproduct of one privacy-preserving query, providing a complete view of that partition without any additional queries.} We call the $t_c$ closest vertices per partition the \emph{TopCand} (top candidates). These vertices are added to the candidate heap $\mathcal{W}$ (min-heap) as additional navigation seeds, enabling jump-style exploration: the search can leap directly to promising regions of a partition without following graph edges step by step, effectively supplementing neighbor-based navigation with partition-level shortcuts. The remaining $K-t_c$ vertices are not promoted to $\mathcal{W}$ (which would overwhelm subsequent \nbrpir rounds), but their distances, obtained for free from the amortized \vecpir query, are still valuable. Like standard graph search, \sysname maintains $\mathcal{W}$ for exploration and a result heap $\mathcal{R}$ (max-heap, capacity $\mathit{ef}$) for the answer set. The graph neighbors decoded from \nbrpir and the $t_c$ \emph{TopCand} vertices enter both $\mathcal{W}$ and $\mathcal{R}$, while the top-$k$ vertices in the partition enter $\mathcal{R}$. These free distances tighten the admission threshold of $\mathcal{R}$ faster, so that fewer candidates qualify in subsequent rounds and fewer PIR queries are wasted.

\subsubsection{Centroid-based pruning}\label{sec:centroid-based pruning}
During neighbor processing, discovered neighbors may point to many uncached partitions. Without prioritization, all would compete for the limited \vecpir batch budget $B_v$, and distant or irrelevant partitions waste valuable slots. The saved partition centroids can naturally serve as directional indicators. We maintain the pending set $\mathcal{P}$ as an ordered set sorted by ascending centroid distance $-\langle \vec{v}_q, \mu_c \rangle$, where $\mu_c$ is the centroid of partition $c$. Since centroids are public metadata, this sorting is performed entirely on the client without online communication. When selecting partitions for the next \vecpir batch, the sorted order ensures that the top-$B_v$ entries are always the closest partitions to the query, directing the communication budget toward the most promising regions. Before selection, the client prunes $\mathcal{P}$ by removing all partitions whose centroid distance exceeds $\alpha \cdot d^*_{\text{best}}$, where $d^*_{\text{best}}$ is the closest distance seen so far\rev{, initialized from the entry partition and updated whenever a closer centroid appears}. \rev{The pruning ratio $\alpha$ is a fixed per-dataset parameter, set once and not updated during the search. A larger $\alpha$ prunes more strictly, which suits datasets whose centroids lie close together, since it still keeps the useful partitions. When the centroids are spread out, a smaller $\alpha$ retains the potential partitions that may hold true neighbors.} Together, sorting and pruning direct the limited \vecpir budget toward the most relevant partitions.

\subsubsection{Selective decryption and caching}
Each \nbrpir query retrieves an entire subgroup column containing $S$ vertices' neighbor data, but the search typically visits only a few vertices per subgroup. Decrypting all $S \cdot M_0 \cdot t$ entries would waste computation. Instead, the client caches the raw ciphertext and performs selective decryption: only the rows of actually visited vertices are decrypted on demand. When later candidates fall into the same subgroup, the cached ciphertext is reused without additional communication.

\subsubsection{Compressed storage}
Server-side answer computation requires multiplying the (potentially multi-GB) PIR database by each query, making memory bandwidth the dominant cost. Since actual database values have small ranges (\vecpir stores uint8-quantized vectors in $[0, 2^8-1]$ and \nbrpir stores $b$-bit parts in $[0, 2^b-1]$), we keep the databases in a compact format and zero-extend each element to the LWE word size at query time. 

\subsection{\sysname Protocols}
\label{sec:search}

We now describe the \sysname protocols: (1) Offline preprocessing (Algorithm~\ref{alg:offline}) builds the PIR databases; (2) Online search (Algorithm~\ref{alg:search}) answers each query with a fixed communication pattern.

\subsubsection{Offline preprocessing}

\begin{algorithm}[t]
\caption{Offline Preprocessing}
\label{alg:offline}
\small
\KwIn{$N$ vectors of dimension $d$, graph degree $M_0$, number of partitions $C$, subgroup size $S$, part width $b$}
\KwOut{PIR databases $\mathsf{db}_{\text{vec}}, \mathsf{db}_{\text{nbr}}$; hint matrices $H_v, H_n$; centroids $\{\mu_c\}$}
\tcc{Index construction}
Build proximity graph $G$ with $M_0$ neighbors per vertex\;
Partition $G$ into $C$ balanced partitions via graph partitioning\;
\ForEach{partition $c$}{
    Divide into subgroups of at most $S$ vertices\;
}
Quantize all vectors to uint8\;
Compute partition centroids $\{\mu_1, \ldots, \mu_C\}$\;
\tcc{PIR database construction}
Build $\mathsf{db}_{\text{vec}}$: horizontal layout, $K \times (C \cdot d)$ (\S\ref{sec:vecpir})\;
Split-encode all neighbors $id$ into $t$ parts of $b$ bits (\S\ref{sec:nbrpir})\;
Build $\mathsf{db}_{\text{nbr}}$: subgroup-column layout, $(S \cdot M_0 \cdot t) \times (C \cdot g_{\max})$\;
\tcc{PIR setup (one-time distribution to client)}
Generate shared random matrices $A_v, A_n$\;
Compute hint matrices: $H_v \gets \mathsf{db}_{\text{vec}} \cdot A_v$, $H_n \gets \mathsf{db}_{\text{nbr}} \cdot A_n$\;
Send $(A_v, H_v, A_n, H_n, \{\mu_c\}_{c=1}^C)$ to $\mathcal{C}$\;
\end{algorithm}

Algorithm~\ref{alg:offline} summarizes server-side preprocessing, executed once before any queries.
Lines~1--5 construct the search index. The server first builds a proximity graph $G$ (e.g., HNSW) where each vertex has $M_0$ neighbors (Line~1), then partitions $G$ into $C$ balanced partitions using graph partitioning (Line~2). Within each partition, vertices are divided into subgroups of at most $S$ vertices (Lines~3--4). All vectors are quantized to uint8 for compact storage (Line~5), and partition centroids are computed as public metadata for centroid-based entry and pruning (Line~6). Then, the server builds the two PIR databases following the layouts described in Figure~\ref{fig:dual-pir}: $\mathsf{db}_{\text{vec}}$ stores quantized vectors in horizontal layout, and $\mathsf{db}_{\text{nbr}}$ stores split-encoded neighbor lists in subgroup columns (Lines~7--9). After that, the server performs the one-time SimplePIR setup. The server generates shared random matrices $A_v, A_n$ and computes hint matrices $H_v, H_n$ by multiplying each database with its corresponding random matrix. These matrices, together with partition centroids, are sent to the client. This offline distribution is amortized over all subsequent queries; the client uses $A$ and $H$ to construct PIR queries and recover answers without additional server interaction (Lines~10--12).

\subsubsection{Online PANNS}

\begin{algorithm}[t]
\caption{Online PANNS}
\label{alg:search}
\small
\KwIn{query vector $\vec{v}_q$, result size $k$, heap capacity $\mathit{ef}$, entry partitions $K_e$, \nbrpir batch size $B_n$, \vecpir batch size $B_v$, iteration count $T$, \emph{TopCand} count $t_c$}
\KwOut{top-$k$ nearest neighbors}
\tcc{Stage 1: Centroid entry}
Compute $\langle \vec{v}_q, \mu_c \rangle$ for all $c \in [C]$ locally\;
Select top-$K_e$ partitions as entry points\;
$\mathsf{BatchVecPir}$($K_e$)\;
Initialize candidate min-heap $\mathcal{W}$, result max-heap $\mathcal{R}$ (size of $\mathit{ef}$), pending sorted set $\mathcal{P}$\;
Insert top-$t_c$ per entry partition into $\mathcal{W}$, insert top-$k$ into $\mathcal{R}$\;
\tcc{Stage 2: Fixed-iteration search loop}
\For{$\emph{iter} = 1$ \KwTo $T$}{
    \tcc{Phase A: Neighbor retrieval}
    Pop candidates from $\mathcal{W}$ until $B_n$ uncached subgroups collected\;
    Pad to $B_n$ with random uncached subgroups\Comment*[r]{dummy}
    $\mathsf{BatchNbrPIR}$($B_n$) and cache ciphertexts\;
    \tcc{Phase B: Neighbor processing}
    \ForEach{popped candidate $v$}{
        Selectively decrypt $v$’s $M_0$ neighbors from cached ciphertext\;
        \ForEach{neighbor $u$ of $v$}{
            \eIf{$u$’s partition is cached}{
                Read $\text{dist}(u, \vec{v}_q)$ from cache\;
                Update $\mathcal{W}$ and $\mathcal{R}$ if $\text{dist}(u,\vec{v}_q) < \mathcal{R}.\mathsf{top}()$\;
            }{
                Insert $u$’s partition into $\mathcal{P}$\;
            }
        }
    }
    \tcc{Phase C: Distance computation}
    Prune $\mathcal{P}$ by centroid-based partition scheduling (\S\ref{sec:opt})\;
    Pop first $B_v$ entries from $\mathcal{P}$, pad to $B_v$\Comment*[r]{dummy}
    $\mathsf{BatchVecPir}$($B_v$)\;
    \ForEach{newly queried partition $c'$}{
        Insert top-$t_c$ vertices into $\mathcal{W}$ (\emph{TopCand})\;
        Insert top-$k$ vertices into $\mathcal{R}$\;
        \ForEach{neighbor $u$ in $\mathcal{P}$ with $u \in c'$}{
            Read $\text{dist}(u, \vec{v}_q)$ from cache\;
            Insert $u$ into $\mathcal{W}$ and $\mathcal{R}$ if $\text{dist}(u,\vec{v}_q) < \mathcal{R}.\mathsf{top}()$\;
        }
    }
}
\Return top-$k$ from $\mathcal{R}$\;
\end{algorithm}

Algorithm~\ref{alg:search} describes the online search executed by the client for a query vector $\vec{v}_q$. The parameters are: the number of results $k$, result heap capacity $\mathit{ef} \geq k$, number of entry partitions $K_e$, batch sizes $B_n$ and $B_v$ for \nbrpir and \vecpir, iteration count $T$, and \emph{TopCand} count per partition $t_c$.

\textbf{Stage~1: Centroid entry} (Lines~1--5). The client computes inner products between query vector $\vec{v}_q$ and all $C$ partition centroids locally, requiring zero online communication since centroids are distributed during offline setup (Line~1). The $K_e$ partitions with the smallest distance are selected as entry points (Line~2), and a single $\mathsf{BatchVecPir}$ query retrieves distances to all vertices in these $K_e$ partitions (Line~3). The candidate heap $\mathcal{W}$ and result heap $\mathcal{R}$ are then initialized (Lines~4--5): the top-$t_c$ closest vertices per entry partition are inserted into $\mathcal{W}$ as initial navigation seeds, and the top-$k$ are inserted into $\mathcal{R}$ (Line~5).

\textbf{Stage~2: Fixed-iteration search loop} (Lines~6--26). Each of the $T$ iterations consists of three phases:

\emph{Phase~A: Neighbor retrieval} (Lines~7--9). The client pops candidates from $\mathcal{W}$ in distance order, collecting up to $B_n$ distinct uncached subgroups that these candidates belong to (Line~7). Candidates whose subgroup is already cached are also popped since their neighbors can be decoded immediately. If fewer than $B_n$ uncached subgroups are needed, the remaining slots are padded with random uncached subgroups to maintain a fixed batch size (Line~8). A single $\mathsf{BatchNbrPIR}$ query retrieves all $B_n$ subgroup columns, and the returned ciphertexts are cached for future reuse (Line~9).

\emph{Phase~B: Neighbor processing} (Lines~10--17). For each popped candidate $v$, the client selectively decrypts only $v$'s $M_0$ neighbor entries from the cached subgroup ciphertext (Line~11). For each decoded neighbor $u$: if $u$'s partition has been previously queried via \vecpir, its distance $\text{dist}(u, \vec{v}_q)$ is already available in the distance cache, and $u$ is inserted into $\mathcal{W}$ and $\mathcal{R}$ if competitive (Lines~13--15); otherwise, $u$'s partition is added to $\mathcal{P}$, a cross-iteration priority queue sorted by centroid distance (Line~17).

\emph{Phase~C: Distance computation} (Lines~18--26). The client first prunes $\mathcal{P}$ using centroid-based filtering (\S\ref{sec:opt}), removing partitions whose centroid distance exceeds the pruning threshold (Line~18). The top-$B_v$ closest partitions are selected for querying, padded if fewer than $B_v$ are available (Line~19). A $\mathsf{BatchVecPir}$ query retrieves distances to all vertices in these $B_v$ partitions (Line~20). For each newly queried partition, the top-$t_c$ vertices are inserted into $\mathcal{W}$ as \emph{TopCand} navigation seeds, and the top-$k$ are inserted into $\mathcal{R}$ (Lines~22--23). Additionally, any neighbors from Phase~B that are deferred because their partition was uncached can now be resolved from the distance cache (Lines~24--26).

After $T$ iterations, the client returns the top-$k$ entries from $\mathcal{R}$ (Line~27). Every ANNS executes exactly $1+T$ \vecpir rounds and $T$ \nbrpir rounds, with fixed batch sizes $B_v$ and $B_n$ per round. The dummy padding in Lines~8 and 19 ensures that the server observes an identical communication pattern regardless of query content. Besides, the LWE-encrypted query vectors are computationally indistinguishable from real queries, so the server cannot determine which slots carry real data and which are padding.

\section{Theoretical Analysis}
\label{sec:analysis}

\subsection{Communication Cost}
\label{sec:comm}

\subsubsection{Online communication}
We analyze the per-query online communication cost. As described in \S\ref{sec:search}, each ANNS executes $1 + T$ \vecpir rounds with batch size $B_v$ and $T$ \nbrpir rounds with batch size $B_n$. We denote the element size as $w = \lceil \log q / 8 \rceil$ bytes.

\emph{Total online communication.}
Each \vecpir round involves an upload of the query matrix and a download of the answer matrix. The $\mathsf{db}_{\text{vec}}$ has $K$ rows and $C \cdot d$ columns, so each round uploads a $(C \cdot d) \times B_v$ query matrix and downloads a $K \times B_v$ answer matrix:
\begin{equation}
\label{eq:comm_emb}
\text{Comm}_{\text{vec}} = (1 + T) \cdot B_v \cdot (\underbrace{C \cdot d}_{\text{upload}} + \underbrace{K}_{\text{download}}) \cdot w
\end{equation}

The $\mathsf{db}_{\mathsf{nbr}}$ has $S \cdot M_0 \cdot t$ rows and $C \cdot g_{\max}$ columns, where $g_{\max} = \lceil K / S \rceil$. Each round uploads a $(C \cdot g_{\max}) \times B_n$ query matrix and downloads an $(S \cdot M_0 \cdot t) \times B_n$ answer matrix:
\begin{equation}
\label{eq:comm_nbr}
\text{Comm}_{\text{nbr}} = T \cdot B_n \cdot (\underbrace{C \cdot \lceil K / S \rceil}_{\text{upload}} + \underbrace{S \cdot M_0 \cdot t}_{\text{download}}) \cdot w
\end{equation}

The total per-query online communication is $\text{Comm}_{\text{online}} = \text{Comm}_{\text{vec}} + \text{Comm}_{\text{nbr}}$.

\emph{Optimal parameter selection.}
The total communication depends on the number of partitions $C$ and the subgroup size $S$, which can be chosen to minimize communication. For \vecpir, the upload cost is proportional to $C \cdot d$ and the download cost is proportional to $K \approx N / C$. Setting the derivative of $C \cdot d + N / C$ with respect to $C$ to zero yields the optimal number of partitions:
\begin{equation}
\label{eq:optC}
C^* = \sqrt{N / d}
\end{equation}
For \nbrpir, the upload cost is proportional to $C \cdot \lceil K / S \rceil \approx N / S$ and the download cost is proportional to $S \cdot M_0 \cdot t$. Setting the derivative of $N / S + S \cdot M_0 \cdot t$ with respect to $S$ to zero yields the optimal subgroup size:
\begin{equation}
\label{eq:optS}
S^* = \sqrt{N / (M_0 \cdot t)}
\end{equation}

\emph{Amortization benefit.}
PIR schemes such as SimplePIR and Piano achieve per-query communication $O(\sqrt{N})$ for a database of $N$ entries. The critical difference between \sysname and per-vertex PIR systems lies in the number of queries. Pacmann issues one PIR query per accessed vertex, so visiting $n_a$ vertices incurs total communication $O(n_a\sqrt{N})$. In \sysname, each \vecpir query retrieves distances to all $K$ vertices in a partition at once, so visiting vertices spread across $c$ distinct partitions requires only $c$ queries with total communication $O(c\sqrt{N})$. Since a typical search accesses multiple vertices within the same partition, $c \ll n_a$, yielding a proportional reduction in communication. \rev{The benefit grows with the dimensionality $d$. A per-vertex cost model scheme downloads the full $d$-dimensional vector of every accessed vertex, so it pays $O(d)$ communication per compared vertex. In \sysname, the optimal $C^*=\sqrt{N/d}$ (Eq.~\ref{eq:optC}) makes both the upload $C{\cdot}d$ and the download $K$ in Eq.~\ref{eq:comm_emb} approximately equal to $\sqrt{N{\cdot}d}$, so one \vecpir query spends about $2\sqrt{N{\cdot}d}$ values on the distances to about $\sqrt{N{\cdot}d}$ vertices, and the amortized cost per compared vertex is a constant independent of $d$.}

\subsubsection{Offline communication}
Before any queries, the client downloads hint matrices and shared random matrices from the server. For \vecpir, the hint matrix $H_v = \mathsf{db}_{\text{vec}} \cdot A_v \in \mathbb{Z}_q^{K \times n}$ and the shared random matrix $A_v \in \mathbb{Z}_q^{(C \cdot d) \times n}$. For \nbrpir, the hint matrix $H_n = \mathsf{db}_{\text{nbr}} \cdot A_n \in \mathbb{Z}_q^{(S \cdot M_0 \cdot t) \times n}$ and the
shared random matrix $A_n \in \mathbb{Z}_q^{(C \cdot g_{\max}) \times n}$. The client also receives the $C$ partition centroids. The total offline communication is:
\begin{equation}
\label{eq:comm_offline}
\text{Comm}_{\text{offline}} = \big(K + C \cdot d + S \cdot M_0 \cdot t + C \cdot g_{\max}\big) \cdot n \cdot w + C \cdot d \cdot w_c
\end{equation}
where $n$ is the LWE dimension and $w_c$ is the per-element size of centroids. This one-time cost is amortized over all subsequent queries and also serves as the client-side storage requirement.

\subsection{Computation Cost}
\label{sec:comp}

The dominant server-side cost is the matrix multiplication required to answer each PIR query. For \vecpir, the server computes $\mathsf{db}_{\text{vec}} \cdot Q \in \mathbb{Z}_q^{K \times B_v}$, which requires $K \cdot (C \cdot d) \cdot B_v$ multiplications per round. For \nbrpir, the server computes $\mathsf{db}_{\text{nbr}} \cdot Q \in \mathbb{Z}_q^{(S \cdot M_0 \cdot t) \times B_n}$, requiring $(S \cdot M_0 \cdot t) \cdot (C \cdot g_{\max}) \cdot B_n$ multiplications per round. The total server-side computation per query is:
\begin{equation}
\label{eq:server_comp}
\text{Comp}_\text{svr} = (1+T) \cdot K \cdot C \cdot d \cdot B_v + T \cdot S \cdot M_0 \cdot t \cdot C \cdot g_{\max} \cdot B_n
\end{equation}
The client-side cost is: generating a query vector requires $O(n)$ operations per column, and recovering the answer requires $O(n)$ operations per row via the precomputed hint.

\emph{Comparison with per-vertex cost model.}
\rev{Let $z$ denote the total size of a vertex's vector and neighbor list. \sysname multiplies each query against the entire database, so one round costs $O(N{\cdot}z)$ and returns information for $O(\sqrt{N})$ vertices. The Piano~\cite{DBLP:conf/sp/ZhouPZS24} PIR used by Pacmann instead aggregates the records of $O(\sqrt{N})$ vertices, so a query costs $O(\sqrt{N}{\cdot}z)$ but returns information for only one vertex. Pacmann is therefore faster per round, and the gap grows with $N{\cdot}z$, but it needs more rounds. The amortized per-vertex computation of the two schemes is comparable. On high-dimensional data, the vector dominates $z$, so the gap is governed by $N{\cdot}d$. \sysname's workload is moreover a dense matrix multiplication that maps directly to GPU and SIMD (AVX-512) acceleration~\cite{DBLP:conf/uss/0004HP022} for great improvement.}

\subsection{Security Analysis}
\label{sec:security}

We prove that \sysname satisfies the query privacy definition given in Definition~\ref{def:privacy} under the LWE assumption.

\begin{theorem}
\label{thm:privacy}
Assuming the hardness of $\mathrm{LWE}_{n, q, \chi}$, \sysname satisfies query privacy (Definition~\ref{def:privacy}) in the semi-honest model.
\end{theorem}

\begin{proof}
We construct a simulator $\mathsf{Sim}$ that, given only the database size $N$ and the public system parameters $(C, S, T, B_v, B_n, K_e)$, produces a transcript indistinguishable from a real execution.

\emph{Offline phase.} The shared random matrices $A_v, A_n$ are public and independent of the query. The hint matrices $H_v = \mathsf{db}_{\text{vec}} \cdot A_v$ and $H_n = \mathsf{db}_{\text{nbr}} \cdot A_n$ are deterministic functions of the public database and the public random matrices, so the simulator outputs them identically to the real protocol.

\emph{Online phase.} In each round, the client sends a query matrix $Q$ whose columns are of the form $qu_j = A \cdot s_j + e_j + \Delta \cdot u_{i_j}$, where $s_j$ is a fresh secret, $e_j$ is drawn from the error distribution $\chi$, and $u_{i_j}$ is a unit vector selecting the target column. By the LWE assumption, $(A, A \cdot s_j + e_j)$ is computationally indistinguishable from $(A, r_j)$ where $r_j \stackrel{\$}{\leftarrow} \mathbb{Z}_q^{C \cdot g_{\max}}$ (or $\mathbb{Z}_q^{C \cdot d}$ for \vecpir). Adding the scaled unit vector $\Delta \cdot u_{i_j}$ does not help the adversary distinguish, because $\Delta \cdot u_{i_j}$ is masked by the LWE noise. Therefore, the simulator can replace each query column with a uniformly random vector $r_j \stackrel{\$}{\leftarrow} \mathbb{Z}_q^{C \cdot g_{\max}}$ (or $\mathbb{Z}_q^{C \cdot d}$ for \vecpir), producing a query matrix indistinguishable from the real one. Since a fresh secret $s_j$ is sampled for each query column, different columns (and different rounds) are independently distributed. Thus, each of the $(1 + T) \cdot B_v$ \vecpir query columns and $T \cdot B_n$ \nbrpir query columns can be simulated independently. The server's response $\mathsf{ans} = \mathsf{db} \cdot Q$ is a deterministic function of the public database and the query matrix. Since the simulator can produce an indistinguishable query matrix, the resulting answer is also indistinguishable.

\emph{Fixed communication pattern.} Every query executes exactly $1 + T$ \vecpir rounds and $T$ \nbrpir rounds, with fixed batch sizes $B_v$ and $B_n$ per round (Algorithm~\ref{alg:search}). When fewer than $B_n$ real subgroups are needed, the client pads with random uncached subgroups. Similarly, when fewer than $B_v$ real partitions are pending, it pads with already-cached partition IDs. Since all query columns are LWE samples indistinguishable from uniform, the server cannot distinguish real query columns from dummy padding columns. The server therefore observes the same number of rounds, the same matrix dimensions, and identically distributed query matrices for every query, regardless of the query content.

Combining the above, the simulated transcript is computationally indistinguishable from the real transcript under the LWE assumption. The server learns only the timing of each communication round. It does not learn which partitions or subgroups are queried (protected by LWE), the query vector, or individual distances. The centroid entry computation is performed entirely on the client side using locally stored centroids and produces no communication.
\end{proof}

\section{Experimental Evaluation}
\label{sec:eval}

\subsection{Setup}

\subsubsection{Datasets}
We use \rev{five} datasets which are chosen to cover diverse scales, dimensionalities, and data distributions. SIFT1M and SIFT10M contain 128-dimensional SIFT~\cite{DBLP:journals/ijcv/Lowe04} descriptors extracted from natural images, sharing the same feature space but differing by $10\times$ in scale, which isolates the effect of dataset size on performance. Deep1M contains 96-dimensional CNN features from the Yandex Deep-1B benchmark~\cite{DBLP:conf/cvpr/YandexL16}, representing a different feature distribution at the 1M scale. MSMARCO~\cite{DBLP:conf/nips/NguyenRSGTMD16} contains 8.8M full 768-dimensional passage vectors from a real-world information retrieval benchmark, representing the setting with high dimensionality and semantic text features. \rev{AGNEWS\footnote{\rev{\url{https://huggingface.co/datasets/vector-index-bench/vibe/blob/main/agnews-pplx-1024-cosine-int8.hdf5}}} contains 769K 1024-dimensional text embeddings from the recent VIBE vector-search benchmark, representing higher-dimensional modern text-embedding models. Accuracy is measured by Recall@10 on SIFT1M, SIFT10M, Deep1M, and AGNEWS, and by MRR@10 on MSMARCO.}


\begin{table*}[t]
\centering
\caption{\rev{Overall comparison between \sysname and Pacmann. \textbf{Bold} indicates the better result; Acc. denotes accuracy; Comp.\ and Comm.\ denote per-query computation and communication.}}
\label{tab:overall comparison}
\small
\setlength{\tabcolsep}{3.5pt}
\renewcommand{\arraystretch}{0.8}
\begin{tabular}{cc c ccc cccc c}
\toprule
\multirow{2}{*}{\textbf{Dataset}} & \multirow{2}{*}{\textbf{Scheme}} & \multirow{2}{*}{\textbf{Acc.}} & \multicolumn{3}{c}{\textbf{Preprocessing}} & \multicolumn{4}{c}{\textbf{Online (per query)}} & \multirow{2}{*}{\textbf{Storage}} \\
\cmidrule(lr){4-6} \cmidrule(lr){7-10}
 & & & Graph & PIR & Comm. & Comm. & Rounds & Comp. & Latency & \\
\midrule
\multirow{2}{*}{SIFT1M}
 & \sysname & \textbf{0.921$\pm$0.005} & 111.8\,s & 6.45\,s & \textbf{92.79\,MB} & \textbf{0.75\,MB} & \textbf{5} & \textbf{38.5$\pm$1.0\,ms} & \textbf{0.289$\pm$0.001\,s} & \textbf{92.79\,MB} \\
 & Pacmann  & 0.901$\pm$0.006 & 78.19\,s & 2.65\,s & 610.35\,MB & 22.89\,MB & 16 & 139.5$\pm$18.2\,ms & 0.855$\pm$0.005\,s & 212.43\,MB \\
\midrule
\multirow{2}{*}{SIFT10M}
 & \sysname & \textbf{0.920$\pm$0.002} & 32.4\,min & 174.95\,s & \textbf{285.27\,MB} & \textbf{3.62\,MB} & \textbf{7} & \textbf{474.0$\pm$11.2\,ms} & \textbf{0.824$\pm$0.011\,s} & \textbf{285.27\,MB} \\
 & Pacmann  & 0.903$\pm$0.007 & 38.09\,min & 29.01\,s & 5.96\,GB & 87.42\,MB & 24 & 680.9$\pm$33.4\,ms & 1.489$\pm$0.023\,s & 776.94\,MB \\
\midrule
\multirow{2}{*}{Deep1M}
 & \sysname & \textbf{0.922$\pm$0.003} & 86.30\,s & 5.57\,s & \textbf{85.80\,MB} & \textbf{1.24\,MB} & \textbf{7} & \textbf{55.4$\pm$1.4\,ms} & \textbf{0.405$\pm$0.001\,s} & \textbf{85.80\,MB} \\
 & Pacmann  & 0.902$\pm$0.003 & 85.39\,s &1.84\,s & 488.28\,MB & 18.46\,MB & 16 & 116.2$\pm$4.5\,ms & 0.849$\pm$0.006\,s & 170.55\,MB \\
\midrule
\multirow{2}{*}{AGNEWS}
 & \sysname & \textbf{0.936$\pm$0.001} & 549.5\,s & 45.48\,s & \textbf{167.23\,MB} & \textbf{0.89\,MB} & \textbf{3} & \textbf{90.7$\pm$5.2\,ms} & \textbf{0.241$\pm$0.005\,s} & \textbf{167.23\,MB}\\
 & Pacmann  & 0.911$\pm$0.011 & 681\,s & 7.05\,s & 3.03\,GB & 82.88\,MB & 10 & 222.8$\pm$11.4\,ms & 0.544$\pm$0.002\,s & 1.10\,GB \\
\midrule
\multirow{2}{*}{MSMARCO}
 & \sysname & \textbf{0.253$\pm$0.007} & 104.88\,min & 724.03\,s & \textbf{476.04\,MB} & \textbf{7.45\,MB} & \textbf{5} & \textbf{1520.8$\pm$15.9\,ms} & 1.771$\pm$0.016\,s & \textbf{476.04\,MB} \\
 & Pacmann  & 0.240$\pm$0.008 & 193.91\,min & 83.48\,s & 26.35\,GB & 404.78\,MB & 24 & 1534.9$\pm$21.1\,ms & \textbf{1.557$\pm$0.012\,s} & 3.60\,GB \\
\bottomrule
\end{tabular}
\end{table*}

\subsubsection{Baselines}
\rev{We compare \sysname against two single-server graph-based PANNS systems that follow the per-vertex cost model, Pacmann~\cite{DBLP:conf/iclr/ZhouSF25} and Compass~\cite{DBLP:conf/osdi/ZhuPZP25}.} Pacmann is the state-of-the-art single-server PIR-based PANNS scheme, which uses Piano PIR~\cite{DBLP:conf/sp/ZhouPZS24} and issues one PIR query per visited vertex during graph traversal. \rev{Compass protects the query with ORAM instead of PIR.} We also compare against Tiptoe~\cite{DBLP:conf/sosp/HenzingerDCZ23}, a cluster-based scheme that achieves single-round search via somewhat homomorphic encryption. \rev{We further compare against Panther~\cite{DBLP:conf/ccs/LiHZ0L0C25}, the most recent cluster-based scheme, which co-designs a batched PIR over the clusters with homomorphic distance evaluation and a garbled-circuit top-$k$.} \rev{\sysname, Pacmann, and Compass accumulate communication over rounds, whereas Tiptoe and Panther are single-round. We vary Panther's budget by the number of clusters it accesses.} Pacmann's communication and computation include the amortized cost of each PIR query, because without amortization it would have to rebuild the hint by downloading the entire database. Pacmann's amortization computation is excluded from latency, since it runs after the result is returned. \rev{For Panther we count only the communication that protects query privacy and omit the part that protects database privacy. We report latency only for \sysname and Pacmann. The server index of Compass and Panther grows to $\ge200$\,GB on several datasets and must be held on disk, which severely degrades their performance.}

\subsubsection{Implementation and environment.}
\sysname is implemented in C++ with OpenMP parallelism. All experiments are conducted on a single machine equipped with an Intel Core i7-10700 CPU @ 2.60\,GHz (8 cores / 16 threads) and 64\,GB RAM, running Ubuntu 20.04 Server. For PIR operations, we use 32-bit LWE ciphertexts with $n=1024$ and $\log q= 32$, \rev{which attains $\lambda\approx128$-bit security, following SimplePIR's parameterization~\cite{DBLP:conf/uss/HenzingerHCMV23}.} We build the proximity graph using HNSW~\cite{DBLP:journals/pami/MalkovY20} with degree $M_0=32$ and construction parameter $\text{efConstruction}=200$, retaining only the level-0 graph. All vectors are scalar-quantized to 8 bits for compact PIR storage. Graph partitioning is performed using METIS~\cite{DBLP:journals/siamsc/KarypisK98} with $C=\lfloor \sqrt{N/d} \rfloor$ partitions. For \nbrpir, each neighbor ID is split-encoded with $b=7$ bits per part and $t=\lceil (\lceil \log_2 C \rceil + \lceil \log_2 K \rceil) / b \rceil$ parts, yielding $t=3$ for 1M-scale datasets and $t=4$ for 10M-scale datasets. The subgroup size $S = \lfloor \sqrt{N/(M_0 {\cdot} t)} \rfloor$. The default search configuration uses $k=10$, $\mathit{ef}=50$, and $K_e = 1$. \rev{The per-dataset search parameters $(B_v, B_n, \alpha, t_c)$ are set to $(2, 2, 0.82, 1)$ for SIFT1M, $(2, 2, 0.78, 1)$ for SIFT10M, $(3, 2, 0.50, 1)$ for Deep1M, $(4, 3, 0.65, 1)$ for MSMARCO, and $(2, 3, 0.50, 2)$ for AGNEWS. We choose the minimum $(B_v, B_n)$ beyond which the accuracy gain is negligible, and the maximum $\alpha$ whose accuracy loss is negligible.} \rev{A search issues one \vecpir query for the entry partitions and then $T$ iterations, each with one \nbrpir and one \vecpir query, so it takes $1+2T$ communication rounds. Increasing $T$ accumulates communication in fixed increments and raises accuracy until it flattens.} Latency is defined as computation time plus communication delay, where each communication round incurs a 50\,ms round-trip time (RTT). \rev{All results are reported as the averages (mean$\pm$std) over 10 repeated executions. }

\subsection{Overall Comparison}
\label{sec:main_results}

Table~\ref{tab:overall comparison} presents the overall comparison between \sysname and Pacmann. \rev{For a fair comparison, we configure \sysname and Pacmann with the minimum number of communication rounds that reaches $\text{Recall@10} = 0.9$ and $\text{MRR@10} = 0.24$}


\subsubsection{Online communication}
\rev{\sysname achieves a 93--99\% reduction in per-query online communication across all datasets: 93\% on Deep1M (1.24\,MB vs.\ 18.46\,MB), 96\% on SIFT1M (0.75\,MB vs.\ 17.17\,MB) and SIFT10M (3.62\,MB vs.\ 87.42\,MB), 98\% on MSMARCO (7.45\,MB vs.\ 404.78\,MB), and 99\% on AGNEWS (0.89\,MB vs.\ 82.88\,MB)}. This improvement stems from the amortization effect: each \vecpir query returns distances to all $K$ vertices in a partition simultaneously, so \sysname retrieves thousands of distances per round while Pacmann retrieves information for only one vertex per PIR query. \rev{The reduction is more pronounced on datasets with higher dimensionality, which matches the amortization analysis in \S\ref{sec:comm}.}

\subsubsection{Accuracy}
Despite using far less communication, \sysname exceeds Pacmann's accuracy on all datasets: \rev{Recall@10 rises by 0.017--0.025 on the four recall datasets (e.g., $0.936\pm0.001$ vs.\ $0.911\pm0.011$ on AGNEWS), and MRR@10 from $0.240\pm0.008$ to $0.253\pm0.007$ on MSMARCO (Table~\ref{tab:overall comparison}). The advantage arises because each amortized \vecpir query returns distances to all vertices in a partition, providing a far richer per-round view than Pacmann's per-vertex retrieval. This richer view enables jump-style navigation such as \emph{TopCand} exploration (\S\ref{sec:opt}) and thus higher accuracy in fewer rounds.}

\begin{figure*}[t]
\centering
\begin{subfigure}[b]{0.19\textwidth}
    \includegraphics[width=\textwidth]{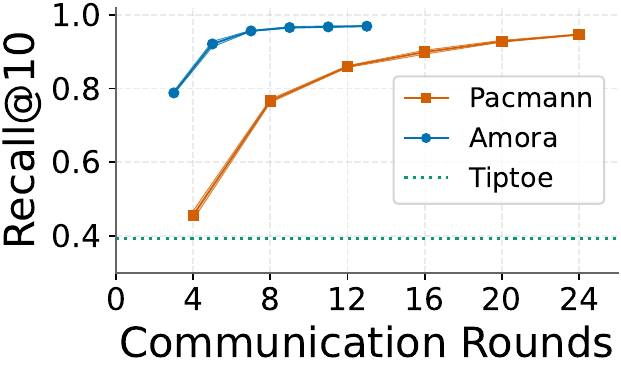}\\
    \includegraphics[width=\textwidth]{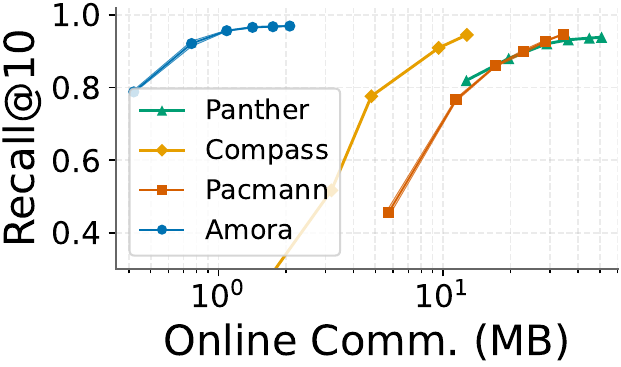}\\
    \includegraphics[width=\textwidth]{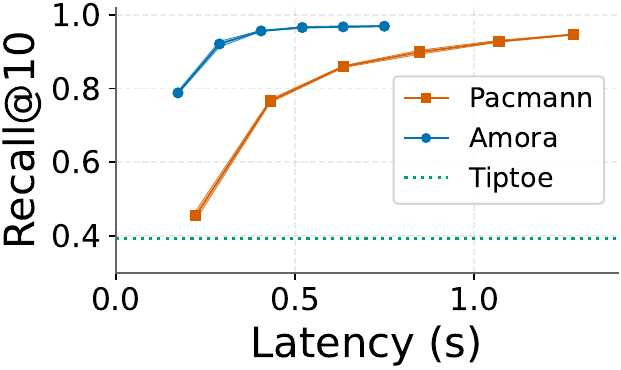}
    \caption{SIFT1M}
\end{subfigure}
\hfill
\begin{subfigure}[b]{0.19\textwidth}
    \includegraphics[width=\textwidth]{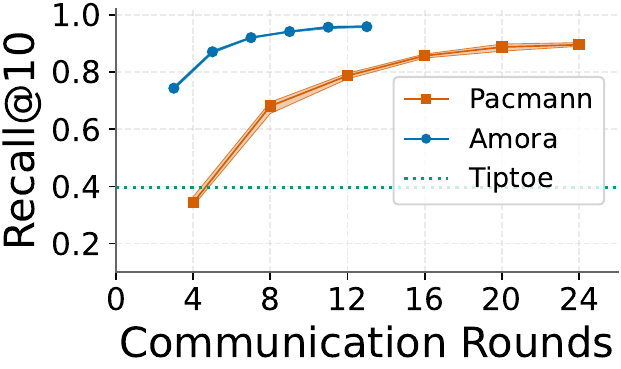}\\
    \includegraphics[width=\textwidth]{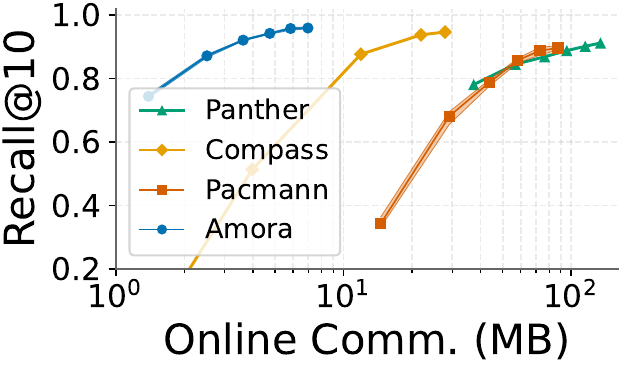}\\
    \includegraphics[width=\textwidth]{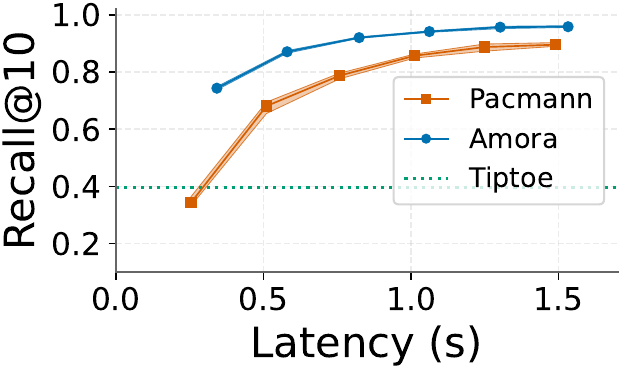}
    \caption{SIFT10M}
\end{subfigure}
\hfill
\begin{subfigure}[b]{0.19\textwidth}
    \includegraphics[width=\textwidth]{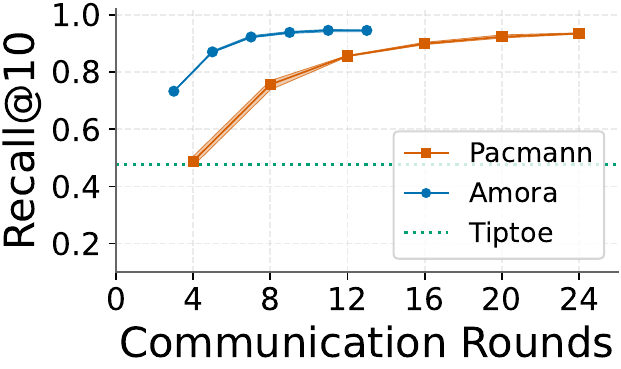}\\
    \includegraphics[width=\textwidth]{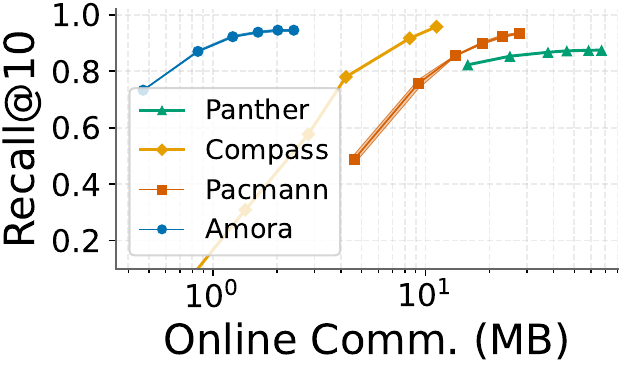}\\
    \includegraphics[width=\textwidth]{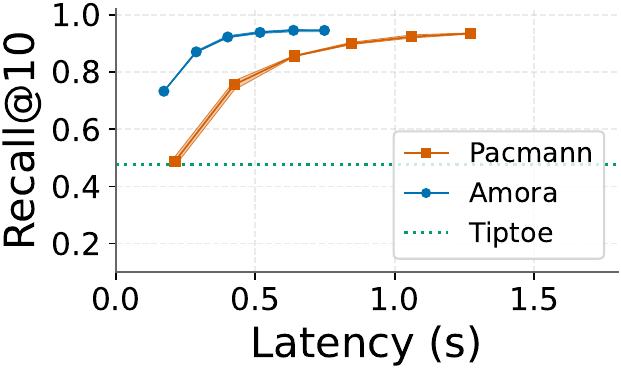}
    \caption{Deep1M}
\end{subfigure}
\hfill
\begin{subfigure}[b]{0.19\textwidth}
    \includegraphics[width=\textwidth]{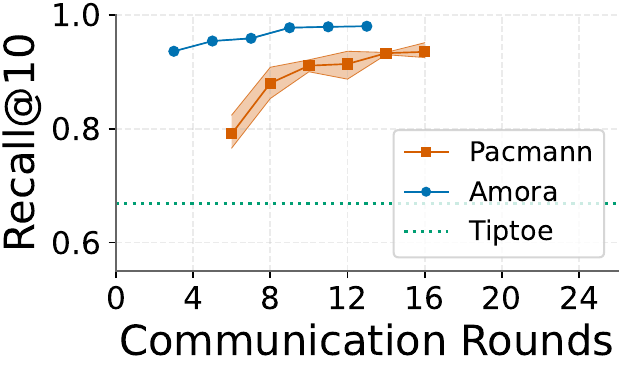}\\
    \includegraphics[width=\textwidth]{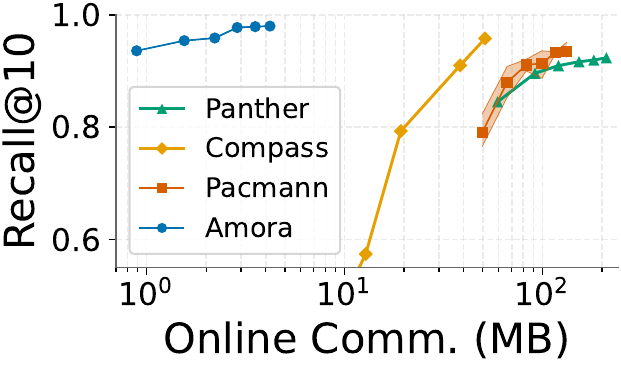}\\
    \includegraphics[width=\textwidth]{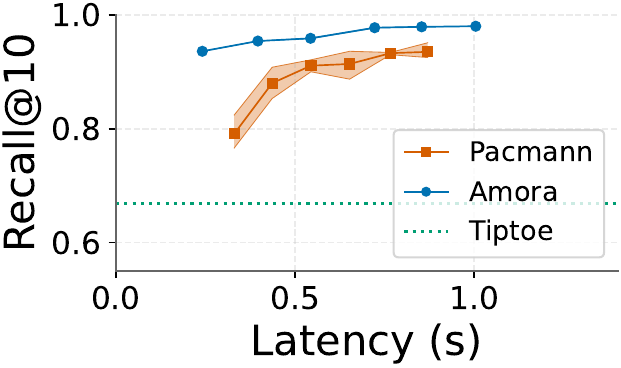}
    \caption{AGNEWS}
\end{subfigure}
\hfill
\begin{subfigure}[b]{0.19\textwidth}
    \includegraphics[width=\textwidth]{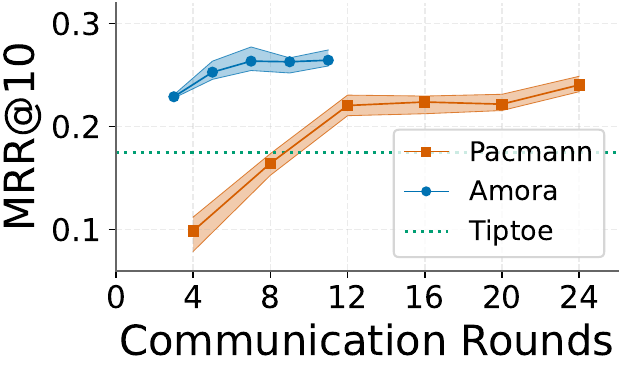}\\
    \includegraphics[width=\textwidth]{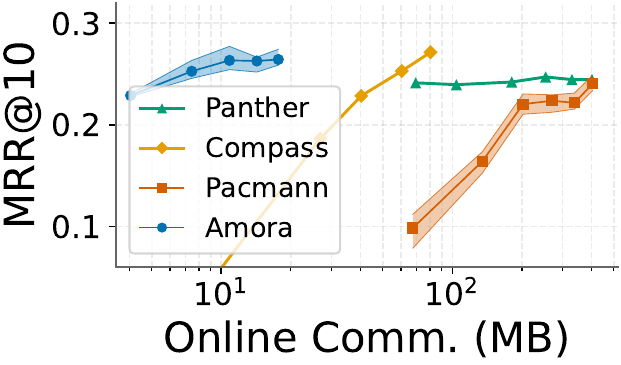}\\
    \includegraphics[width=\textwidth]{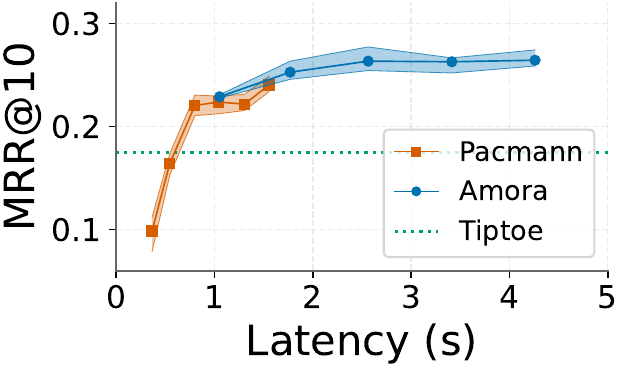}
    \caption{MSMARCO}
\end{subfigure}
\caption{\rev{Accuracy vs.\ communication rounds (top), online communication (middle), and latency (bottom) per query.}}
\label{fig:recall-rounds}
\label{fig:recall-comm}
\label{fig:recall-latency}
\end{figure*}

\subsubsection{Preprocessing and storage}
\sysname's offline communication (hint matrices + shared random matrices) is 82--98\% smaller than Pacmann's, since Pacmann requires downloading the entire database for Piano's offline phase. Client-side storage exhibits the same advantage, as \sysname stores only the hint and the centroids of all partitions rather than a database replica.

\subsubsection{Computation and latency}
\rev{\sysname achieves 45--66\% lower latency at matched accuracy on SIFT1M ($0.289\pm0.001$\,s vs.\ $0.855\pm0.005$\,s), Deep1M ($0.405\pm0.001$\,s vs.\ $0.849\pm0.006$\,s), SIFT10M ($0.824\pm0.011$\,s vs.\ $1.489\pm0.023$\,s), and AGNEWS ($0.241\pm0.005$\,s vs.\ $0.544\pm0.002$\,s), because the reduction in communication rounds more than compensates for the SimplePIR computation overhead.} On MSMARCO, however, \sysname's latency (1.771\,s) slightly exceeds Pacmann's (1.557\,s). \rev{As \S\ref{sec:comp} shows, the computation gap grows with $N{\cdot}d$, so the crossover between the two schemes is governed by this product. MSMARCO has the largest $N{\cdot}d$ among the five datasets ($\approx6.8{\times}10^{9}$), while \sysname remains faster on AGNEWS ($N{\cdot}d\approx7.9{\times}10^{8}$) and SIFT10M ($N{\cdot}d\approx1.3{\times}10^{9}$), so the crossover lies between $1.3{\times}10^{9}$ and $6.8{\times}10^{9}$. The small latency gap on MSMARCO further suggests that the crossover is close to $6.8{\times}10^{9}$. The crossover also depends on the network delay. Because \sysname saves $19$ rounds per query on MSMARCO against a latency gap of $0.21$\,s, it becomes faster again once the RTT exceeds about $61$\,ms.}

\textbf{Scalability discussion.} Table~\ref{tab:overall comparison} also reveals how the costs scale. First, they grow sublinearly with $N$: from SIFT1M to SIFT10M ($10\times$ vectors), per-query communication at the same 7 rounds grows from 1.09\,MB to 3.62\,MB ($3.3\times$), and hint storage from 92.79\,MB to 285.27\,MB ($3.1\times$), matching the $O(\sqrt{N{\cdot}d})$ per-round cost in \S\ref{sec:comm}. Second, dimensionality shifts the cost composition without weakening the advantage: at the 1M scale, communication is dominated by the $d$-independent \nbrpir, so SIFT1M and Deep1M cost almost the same (1.09\,MB vs.\ 1.24\,MB at 7 rounds); on MSMARCO ($d=768$) and AGNEWS ($d=1024$), the wide $\mathsf{db}_{\text{vec}}$ columns become dominant (7.45\,MB and 0.89\,MB per query), yet the reductions over Pacmann are the largest of the five datasets (98\% and 99\%).

\subsection{\rev{Accuracy vs.\ Online Communication}}
\label{sec:recall_rounds}

Figure~\ref{fig:recall-rounds} shows how accuracy improves as the number of communication rounds increases \rev{(top row) and as online communication accumulates (middle row). }

\subsubsection{\rev{\sysname surpasses both graph-based baselines with far less communication}}
\label{sec:tradeoff_comm}
\rev{\sysname surpasses Pacmann's 24-round accuracy within 7 rounds on SIFT1M ($0.956{\pm}0.001$ vs.\ $0.947{\pm}0.002$), SIFT10M ($0.920{\pm}0.002$ vs.\ $0.903{\pm}0.007$), and MSMARCO ($0.263{\pm}0.009$ vs.\ $0.240{\pm}0.008$), and within 11 rounds on Deep1M ($0.946{\pm}0.003$ vs.\ $0.935{\pm}0.002$). On AGNEWS, 3 rounds already match Pacmann's 16-round result ($0.936{\pm}0.001$ vs.\ $0.935{\pm}0.014$). \sysname reaches these points with $93\%$--$99\%$ less online communication, for example 1.09\,MB vs.\ 34.34\,MB on SIFT1M and 0.89\,MB vs.\ 132.61\,MB on AGNEWS. The gap is not tied to one operating point. In the middle row of Figure~\ref{fig:recall-comm}, \sysname's most expensive setting still costs less than Pacmann's cheapest one on every dataset, 2.08\,MB vs.\ 5.72\,MB on SIFT1M and 4.20\,MB vs.\ 49.73\,MB on AGNEWS.} This gap stems from the amortization effect: each \vecpir query returns distances to all $K$ vertices in a partition simultaneously, providing a much denser view of the search space per round than Pacmann's per-vertex  retrieval. \rev{At matched accuracy, \sysname saves $77\%$--$98\%$ of Compass's online communication, for example 3.62\,MB vs.\ 15.92\,MB on SIFT10M and 0.89\,MB vs.\ 44.90\,MB on AGNEWS. Compass also follows the per-vertex cost model, and Ring ORAM pads every bucket it touches with dummy blocks, so much of each transfer carries no useful data. Because each block holds a full vector, the saving grows with the dimension at a comparable scale, from $86\%$ on Deep1M ($d{=}96$) to $93\%$ on SIFT1M ($d{=}128$) and $98\%$ on AGNEWS ($d{=}1024$), while the two larger datasets fall to SIFT10M's $77\%$ and MSMARCO's $86\%$.}

\subsubsection{Convergence speed is related to scale and dimensionality}
\rev{SIFT1M ($d=128$) reaches $0.921\pm0.005$ at 5 rounds and $0.967\pm0.002$ at 11 rounds, while Deep1M ($d=96$) starts slower ($0.871\pm0.003$ at 5 rounds) and reaches $0.946\pm0.003$ at 11 rounds.} Scaling from SIFT1M to SIFT10M ($10\times$ more vectors, same $d=128$) slows convergence to $0.871\pm0.004$ at 5 rounds, as the larger graph requires more partitions to be explored. In fact, scale and dimensionality act through the partition count $C^*=\sqrt{N/d}$ (Eq.~\ref{eq:optC}): each \vecpir query reveals one partition, so a smaller $N$ or a larger $d$ means fewer partitions and faster convergence. \rev{AGNEWS converges fastest ($C^*=27$, $0.936{\pm}0.001$ at 3 rounds), while Deep1M ($C^*=102$) and SIFT10M ($C^*=279$) are the slowest ($0.733{\pm}0.001$ and $0.744{\pm}0.004$ at the same budget). MSMARCO ($C^*=107$) also gains little after 7 rounds ($0.263\pm0.009$ at 7 rounds vs.\ $0.264\pm0.006$ at 11).}

\subsubsection{Iterative graph traversal outperforms single-round retrieval}
Tiptoe's single-round cluster-based approach achieves substantially lower accuracy on all five datasets: 0.39 on SIFT1M, 0.48 on Deep1M, 0.40 on SIFT10M, 0.67 on AGNEWS, and 0.17 on MSMARCO, all below \sysname at only 3 rounds. \rev{We enlarge Panther's communication budget by accessing more clusters to achieve a higher accuracy. On the four recall datasets, \sysname saves $97\%$--$99\%$ of the communication and improves accuracy, e.g., 0.75\,MB vs.\ 26.2\,MB at $0.921{\pm}0.005$ vs.\ 0.907 on SIFT1M, and 0.89\,MB vs.\ 120.5\,MB at $0.936{\pm}0.001$ vs.\ 0.910 on AGNEWS. On MSMARCO, \sysname saves $97\%$ of the communication (7.45\,MB vs.\ 253.3\,MB) and improves MRR@10 by 0.006 ($0.253{\pm}0.007$ vs.\ 0.247).} The reason is that cluster-based retrieval scans only a small number of nearby clusters without iterative refinement across partition boundaries, so nearest neighbors residing in other partitions are ignored.

\subsection{Accuracy vs.\ Latency}
\label{sec:tradeoff_latency}

\rev{The bottom row of Figure~\ref{fig:recall-latency}} compares the accuracy-latency tradeoff by varying the number of rounds.


\subsubsection{\rev{Fewer rounds translate into lower latency}}
\rev{Figure~\ref{fig:recall-latency} shows that \sysname's accuracy results flatten early: by 0.41\,s on SIFT1M, 0.52\,s on Deep1M, 1.06\,s on SIFT10M, and 0.72\,s on AGNEWS, the remaining accuracy gain is within about 0.02. At 0.85, 0.84, 1.01, and 0.54\,s, Pacmann has reached only $0.899{\pm}0.006$, $0.899{\pm}0.004$, $0.857{\pm}0.007$, and $0.911{\pm}0.009$, versus \sysname's $0.956{\pm}0.001$, $0.939{\pm}0.003$, $0.942{\pm}0.002$, and $0.977{\pm}0.001$ at the latencies above. The latency standard deviation stays below $0.05$\,s at all these points. On these datasets, computation accounts for 13\%, 14\%, 58\%, and 38\% of \sysname's per-query latency, and for merely 6\%, 6\%, 19\%, and 8\% of Pacmann's; the remainder is network round-trip delay. Latency is therefore governed by the round count, even on SIFT10M where computation takes over half of \sysname's latency, and \sysname's fewer rounds translate directly into the 45--66\% latency reductions.}

\subsubsection{\rev{Computation dominates latency on MSMARCO}}
\rev{MSMARCO's product of scale and dimensionality is the largest among the five datasets ($N{\cdot}d\approx6.8{\times}10^{9}$), making each matrix multiplication expensive: server computation reaches $1.521\pm0.016$\,s at 5 rounds, accounting for 86\% of total latency, so each additional iteration costs \sysname 0.62--0.75\,s of computation, while Pacmann spends only about 0.06\,s per round, mostly on the RTT. Even so, \sysname matches Pacmann at the low-latency end ($\approx 1.05$\,s) and reaches higher accuracy ($0.253\pm0.007$ vs.\ $0.240\pm0.008$). }

\begin{figure*}[t]
\centering
\begin{subfigure}[b]{0.195\textwidth}
    \includegraphics[width=\textwidth]{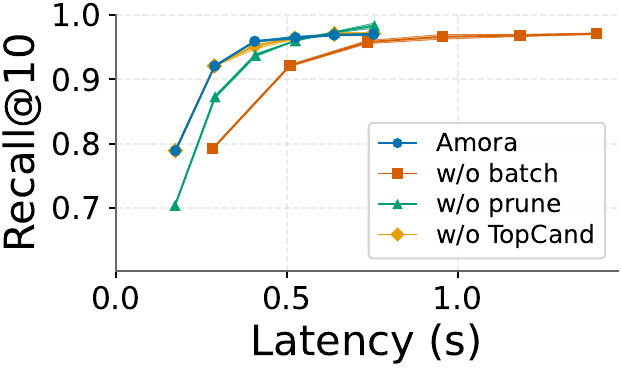}
    \caption{SIFT1M}
\end{subfigure}
\hfill
\begin{subfigure}[b]{0.195\textwidth}
    \includegraphics[width=\textwidth]{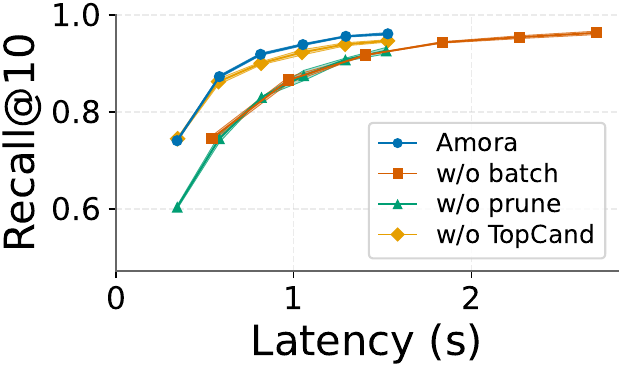}
    \caption{SIFT10M}
\end{subfigure}
\hfill
\begin{subfigure}[b]{0.195\textwidth}
    \includegraphics[width=\textwidth]{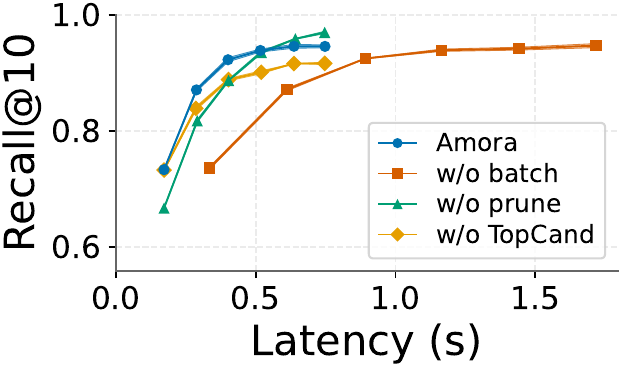}
    \caption{Deep1M}
\end{subfigure}
\hfill
\begin{subfigure}[b]{0.195\textwidth}
    \includegraphics[width=\textwidth]{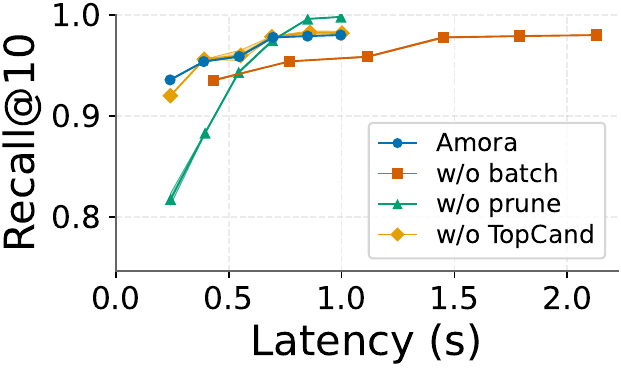}
    \caption{AGNEWS}
\end{subfigure}
\hfill
\begin{subfigure}[b]{0.195\textwidth}
    \includegraphics[width=\textwidth]{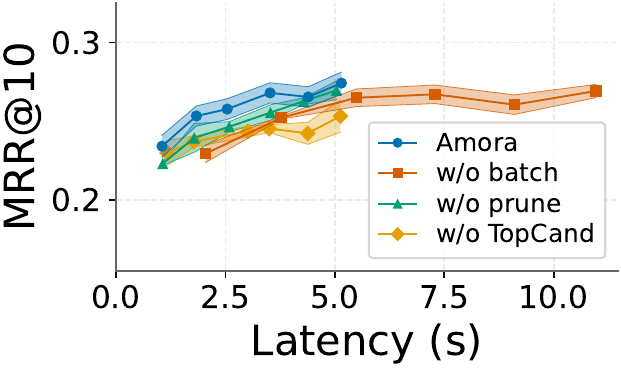}
    \caption{MSMARCO}
\end{subfigure}
\caption{\rev{Optimization techniques ablation study.}}
\label{fig:ablation}
\end{figure*}

\subsection{Ablation Study}
\label{sec:ablation}

\subsubsection{Optimization techniques}
We ablate three key optimization techniques described in \S\ref{sec:opt} to quantify their individual contributions. For each ablation, we disable one technique (denoted "w/o") while keeping all others enabled. \rev{Figure~\ref{fig:ablation} reports accuracy against latency as the iteration count $T$ grows.}

\noindent\textbf{Batch PIR reduces latency.}
Without batching, each PIR query occupies its own communication round, so one search iteration costs $B_v+B_n$ rounds instead of two. At the operating points in Table~\ref{tab:overall comparison}, the unbatched search needs 9 rounds on SIFT1M, 13 on SIFT10M, 16 on Deep1M, 6 on AGNEWS, and 15 on MSMARCO, adding proportionally more RTT delays. Each unbatched round also requires an independent full database scan, whereas batching amortizes the scan across multiple queries in a single wider matrix multiplication and consumes less memory bandwidth. Across all datasets and round counts, disabling batching raises the computation by $1.4$--$2.0\times$, and at the operating points it grows from 0.039\,s to 0.059\,s on SIFT1M and from 1.588\,s to 3.022\,s on MSMARCO. Figure~\ref{fig:ablation} shows that the "w/o batch" results track \sysname in accuracy but are shifted right along the latency axis. At matched accuracy, the latency rises by $1.6$--$2.3\times$, for example from 0.29\,s to 0.51\,s on SIFT1M and from 1.84\,s to 3.77\,s on MSMARCO. All standard deviations of computation and latency here are below $0.02$\,s.

\noindent\textbf{\emph{TopCand} exploration improves accuracy under a tight budget.}
Disabling \emph{TopCand} exploration forces the search to follow graph edges only. At the operating points in Table~\ref{tab:overall comparison}, \rev{Recall@10 drops by $0.034{\pm}0.004$ on Deep1M, $0.018{\pm}0.004$ on SIFT10M, and $0.016{\pm}0.001$ on AGNEWS, and MRR@10 drops by $0.016{\pm}0.007$ on MSMARCO, while on SIFT1M the gap stays within 0.007 at every round count.} \emph{TopCand} exploration seeds the candidate heap with each partition's closest vertices, so the search reaches promising regions that the candidates' edges cannot within the round budget.

\noindent\textbf{Centroid-based pruning improves accuracy under a tight budget.}
Disabling centroid-based pruning reverts to a first-come-first-served ordering with no pruning. At the operating points in Table~\ref{tab:overall comparison}, accuracy falls on all five datasets: \rev{Recall@10 by $0.118{\pm}0.006$ on AGNEWS, $0.088{\pm}0.002$ on SIFT10M, $0.048{\pm}0.002$ on SIFT1M, and $0.036{\pm}0.004$ on Deep1M, and MRR@10 by $0.014{\pm}0.011$ on MSMARCO.} Each round admits only $B_v$ partitions, and centroid-based pruning directs this limited budget toward the closest partitions rather than toward whichever ones are discovered first. The gain is therefore largest under \sysname's tight round budgets.

\subsubsection{\rev{Dual-PIR split layout}}
\label{sec:combined-ablation}
\rev{We validate the necessity of the dual-PIR split against a combined layout that stores the vector and the neighbor list of each vertex in the same column (Figure~\ref{fig:split}). We build the combined database with its own optimal column capacity $S_c= \sqrt{N/(d+M_0\,t)}$, and we drive it with the same centroid entry and graph navigation as \sysname. We compare the two layouts under two conditions. First, we give the combined layout the same online communication budget that the split layout spends at its operating point in Table~\ref{tab:overall comparison}. Within this budget, the accuracy of the combined layout collapses to $0.004$ on SIFT1M, $0.044$ on SIFT10M, $0.027$ on Deep1M, and $0$ on MSMARCO, and one combined query on AGNEWS already exceeds the whole budget. Second, we measure the cost for the combined layout to reach the same accuracy as the split layout, and Figure~\ref{fig:combined} reports the results. The combined layout pays $25\times$ to $161\times$ more online communication and $58\times$ to $286\times$ more computation. The gap comes from the heavy per-vertex payload. A combined column carries $d+M_0\,t$ values per vertex, so one query returns only a few vertices and the search issues far more queries. The dual-PIR split is therefore a necessary design.}

\begin{figure}[t]
\centering
\includegraphics[width=0.7\columnwidth]{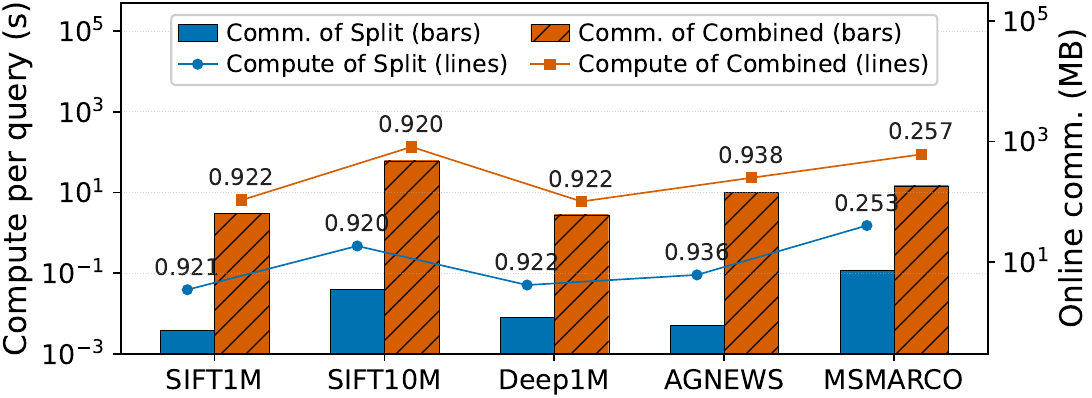}
\caption{\rev{Combined layout vs.\ dual-PIR split layout.}}
\label{fig:combined}
\end{figure}

\subsection{\rev{Parameter Sensitivity Study}}
\label{sec:sensitivity}

\subsubsection{\rev{Result-heap capacity $\mathit{ef}$}}
\label{sec:ef}

\rev{Figure~\ref{fig:ef} sweeps $\mathit{ef}$ from $k=10$ to $300$ on all five datasets at their operating points in Table~\ref{tab:overall comparison}. In Figure~\ref{fig:ef}(a), accuracy rises slightly from $\mathit{ef}=k$ and then stays flat. The largest gain is on SIFT10M, where Recall@10 grows from $0.895{\pm}0.002$ at $\mathit{ef}=10$ to $0.920{\pm}0.002$ at $\mathit{ef}=50$. Over the same range, the gain is $0.012{\pm}0.006$ on MSMARCO, $0.009{\pm}0.005$ on Deep1M, $0.004{\pm}0.001$ on SIFT1M, and below $0.001$ on AGNEWS. Beyond $\mathit{ef}=50$ and up to $\mathit{ef}=300$, no dataset moves by more than $0.007$, which is at the level of the standard deviation across the five repeated runs ($0.001$--$0.008$). In Figure~\ref{fig:ef}(b), the communication is exactly constant on each dataset (0.76, 3.62, 1.24, 0.89, and 7.45\,MB) and the latency stays flat, because the number of PIR queries is fixed by $T$, $B_v$, and $B_n$. The accuracy trend follows from the role of $\mathit{ef}$ as the admission threshold of the result heap. When $\mathit{ef}$ is close to $k$, the tight threshold rejects some promising candidates, and the \nbrpir batch is padded with random subgroups, so part of the budget is wasted and accuracy drops slightly. Once $\mathit{ef}$ is large enough for the closest candidates to fill $B_n$ in every round, a larger $\mathit{ef}$ only admits distant vertices that the search never pops, so the explored candidates and the accuracy stay the same. The number of popped and expanded candidate vertices confirms this. From $\mathit{ef}=10$ to $30$ the count rises from 122 to 131 on SIFT1M, 180 to 201 on SIFT10M, 179 to 203 on Deep1M, and 199 to 268 on MSMARCO, and then changes by less than $4\%$ up to $\mathit{ef}=300$. On AGNEWS it stays at 62 throughout, so the accuracy does not move. We therefore set $\mathit{ef}=50$, safely inside the flat range.}

\subsubsection{\rev{TopCand count $t_c$}}
\label{sec:tc}
\rev{Figure~\ref{fig:tc} sweeps $t_c{\in}\{1,2,4,8\}$ on all five datasets at their operating points in Table~\ref{tab:overall comparison}, with all other parameters fixed. The communication is identical across the sweep, because the number of PIR queries is fixed by $T$, $B_v$, and $B_n$. Therefore, $t_c$ affects only the accuracy. On SIFT1M, SIFT10M, and Deep1M, the accuracy declines as $t_c$ grows, with Recall@10 on SIFT1M dropping from $0.921{\pm}0.005$ at $t_c=1$ to $0.866{\pm}0.001$ at $t_c=8$. On AGNEWS, $t_c=2$ performs best ($0.936{\pm}0.001$ vs.\ $0.919{\pm}0.001$ at $t_c=1$), and on MSMARCO $t_c$ has little effect on the MRR@10, which stays between $0.239$ and $0.247$. Competition inside the candidate heap explains why a large $t_c$ degrades the accuracy. The \emph{TopCand} seeds and the graph neighbors share the heap, and each iteration pops only $B_n$ candidates for the \nbrpir batch. Seeds from the same partition lie close to one another and bring little new information, so a large $t_c$ fills the limited slots with redundant seeds and displaces the edge expansion around the best candidates. On AGNEWS, the single-iteration budget ($T=1$) limits how far the edge expansion can proceed, so a second seed broadens the exploration. We therefore set $t_c=2$ on AGNEWS and $t_c=1$ on the other four datasets.}

\begin{figure}[t]
\centering
\newlength{\efpanelw}\setlength{\efpanelw}{0.34\textwidth}
\begin{subfigure}[b]{\columnwidth}
    \centering
    \includegraphics[width=\efpanelw]{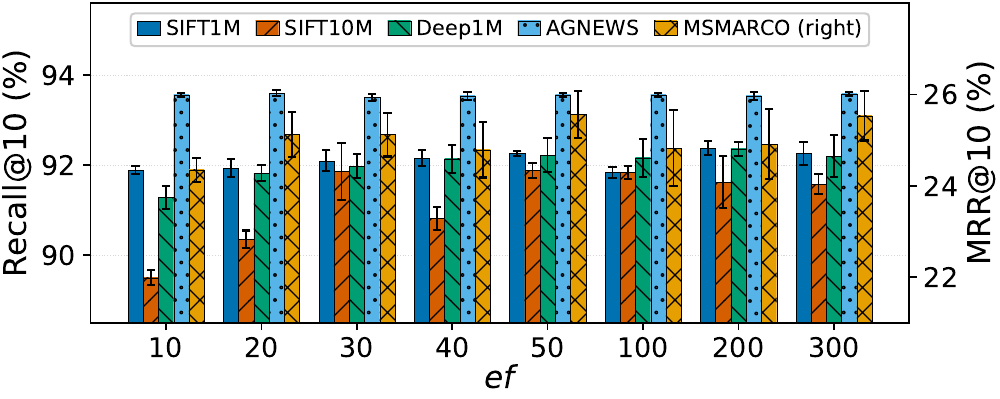}
    \caption{Accuracy.}
\end{subfigure}\\[2pt]
\begin{subfigure}[b]{\columnwidth}
    \centering
    \includegraphics[width=\efpanelw]{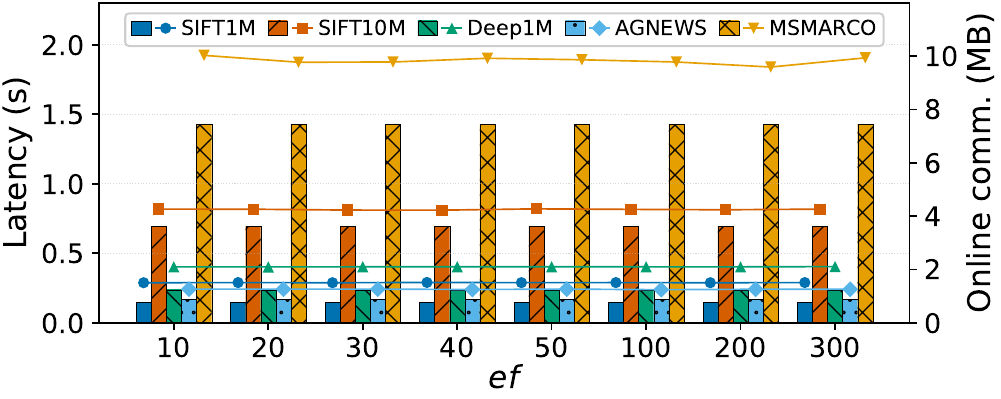}
    \caption{Online communication (bars) and latency (lines).}
\end{subfigure}
\caption{\rev{Sensitivity of the result-heap capacity $\mathit{ef}$.}}
\label{fig:ef}
\end{figure}

\begin{figure}[t]
\centering
\includegraphics[width=0.35\textwidth]{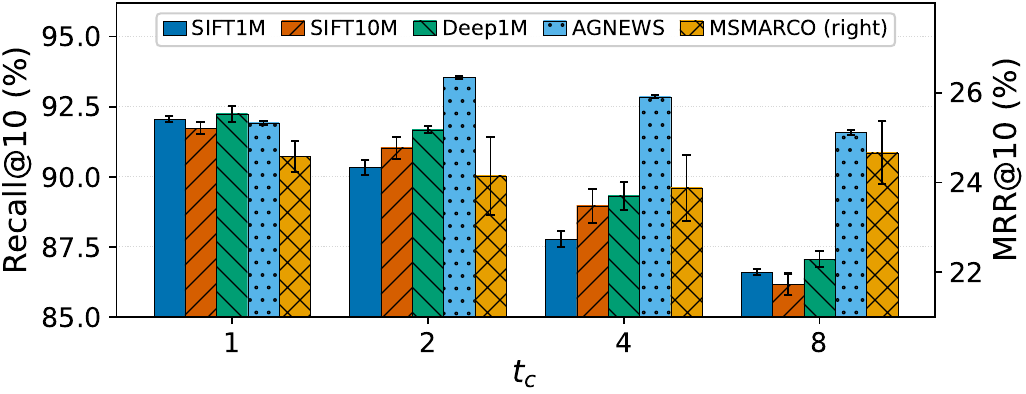}
\caption{\rev{Sensitivity of the \emph{TopCand} count $t_c$.}}
\label{fig:tc}
\end{figure}

\subsubsection{\rev{Partition count $C$}}
\label{sec:csens}
\rev{Figure~\ref{fig:csens} sweeps the partition count over $\{C^*/4, C^*/2, C^*, 2C^*, 4C^*\}$ on all five datasets, where $C^*=\lfloor\sqrt{N/d}\rfloor$ is the default from Eq.~\ref{eq:optC}. Each level uses a rebuilt index and a unified search protocol ($T=3$, $t_c=1$). Figure~\ref{fig:csens}(a) shows that the communication reaches its minimum at $C^*$ on every dataset, and both $C^*/4$ and $4C^*$ pay $1.6\times$ to $1.9\times$ the communication of $C^*$. Figure~\ref{fig:csens}(b) shows that increasing $C$ beyond $C^*$ lowers the accuracy on every dataset, from $C^*$ to $4C^*$ by $0.120$ on Deep1M, $0.103$ on SIFT10M, $0.067$ on SIFT1M, and $0.022$ on both AGNEWS and MSMARCO, because smaller partitions let one \vecpir query return fewer distances. Decreasing $C$ below $C^*$ helps only on some datasets: at $C^*/4$ the accuracy rises on SIFT10M ($+0.054$), MSMARCO ($+0.009$), and SIFT1M ($+0.002$), but falls on Deep1M ($-0.009$) and AGNEWS ($-0.005$). We therefore use $C^*$ on all datasets, since it minimizes the communication and achieves a high level of accuracy.}

\begin{figure}[t]
\centering
\begin{subfigure}[b]{0.464\columnwidth}
    \centering
    \includegraphics[width=\textwidth]{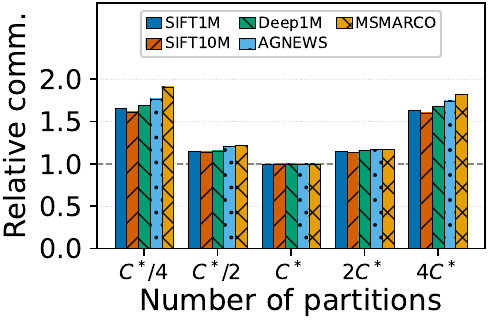}
    \caption{Communication.}
\end{subfigure}%
\begin{subfigure}[b]{0.536\columnwidth}
    \centering
    \includegraphics[width=\textwidth]{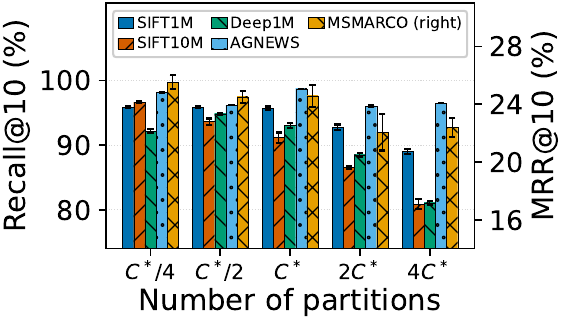}
    \caption{Accuracy.}
\end{subfigure}
\caption{\rev{Sensitivity of the partition count $C$.}}
\label{fig:csens}
\end{figure}





\section{Conclusion}
\label{sec:conclusion}

We presented \sysname, a communication-efficient PANNS system with high accuracy. The core is a dual-PIR split layout built on SimplePIR and graph partitioning: \vecpir retrieves distances to all vertices in a partition at once, and \nbrpir retrieves neighbor lists at subgroup granularity, maximizing the information retrieved per PIR query. Based on the dual-PIR, batch PIR merges concurrent queries to reduce communication rounds and the number of database scans, centroid-based pruning directs the limited PIR budget toward the most promising regions, and \emph{TopCand} exploration seeds jump-style graph navigation from each retrieved partition. Experiments on \rev{five} datasets of different scales and dimensionalities show that \sysname reduces online communication by \rev{93--99\%} over the competitors while achieving higher search accuracy.


\bibliographystyle{ACM-Reference-Format}
\bibliography{references}

\end{document}